\documentclass[conference]{IEEEtran}

\usepackage{graphicx} 
\usepackage{balance}
\usepackage{amsmath}
\usepackage{amssymb}
\usepackage{amsfonts}
\usepackage{amsthm}
\usepackage{color}
\usepackage{wrapfig}
\usepackage{epsfig}
\usepackage{xspace}
\usepackage{relsize}
\usepackage{colortbl}
\usepackage[update,prepend]{epstopdf}
\usepackage{subfigure,color}
\usepackage[font=small,justification=centering,singlelinecheck=off]{caption}

\newtheorem{example}{\textbf{Example}}

\newtheorem{theorem}{\textbf{Theorem}}
\newtheorem{lemma}{\textbf{Lemma}}

\newtheorem{definition}{Definition}

\makeatletter
\newif\if@restonecol
\makeatother

\usepackage[noend]{algpseudocode}
\usepackage[linesnumbered,ruled,vlined, noend]{algorithm2e}
\usepackage[noadjust]{cite}
\usepackage[table]{xcolor}
\usepackage{multirow}

\usepackage[bookmarks=false,hidelinks]{hyperref}

\newcommand{\kwnospace}[1]{{\ensuremath {\mathsf{#1}}}}
\newlength{\figsize} 

\newcommand{\degree} {deg}

\newcommand{\abcore} {($\alpha$, $\beta$)\textnormal{-}core\xspace}
\newcommand{\abcores} {($\alpha$, $\beta$)\textnormal{-}cores\xspace}
\newcommand{\abcoreg} {R_{\alpha,\beta}\xspace}

\newcommand{\abcorec} {($\alpha$, $\beta$)\textnormal{-}connected component\xspace}
\newcommand{\abcorecs} {($\alpha$, $\beta$)\textnormal{-}connected components\xspace}
\newcommand{\abcorecom} {($\alpha$, $\beta$)\textnormal{-}community\xspace}
\newcommand{\abcorecoms} {($\alpha$, $\beta$)\textnormal{-}communities\xspace}
\newcommand{\abcorecg} {C_{\alpha,\beta}\xspace}
\newcommand{\abcorecgq} {C_{\alpha,\beta}(q)\xspace}
\newcommand{\wcomg} {\mathcal{R}\xspace}

\newcommand{\wcom} {significant ($\alpha$, $\beta$)-community\xspace}
\newcommand{\wcoms} {significant ($\alpha$, $\beta$)-communities\xspace}

\newcommand{\abindex} {bicore index\xspace}
\newcommand{\indexbsa} {I_{bs}^{\alpha}\xspace}
\newcommand{\indexbsb} {I_{bs}^{\beta}\xspace}
\newcommand{\indexbsas} {\kwnospace{size}(I_{bs}^{\alpha})\xspace}

\newcommand{\indexbsat} {\kwnospace{TC}(I_{bs}^{\alpha})\xspace}

\newcommand{\indexad} {I_{\delta}\xspace}

\newcommand{\indexads} {\kwnospace{size}(I_{\delta})\xspace}
\newcommand{\indexada} {I_{\delta}^{\alpha}\xspace}
\newcommand{\indexadb} {I_{\delta}^{\beta}\xspace}

\newcommand{\indexbc} {I_{v}\xspace}

\newcommand{\size}{\kwnospace{size}\xspace}
\newcommand{\sort}{\kwnospace{sort}\xspace}

\newcommand{\peel} {\kwnospace{SCS}\textrm{-}\kwnospace{Peel}\xspace}
\newcommand{\expand} {\kwnospace{SCS}\textrm{-}\kwnospace{Expand}\xspace}
\newcommand{\binary} {\kwnospace{SCS}\textrm{-}\kwnospace{Binary}\xspace}
\newcommand{\baseline} {\kwnospace{SCS}\textrm{-}\kwnospace{Baseline}\xspace}

\renewcommand{\baselinestretch}{0.96}
\usepackage{enumitem}
\setlist{nolistsep}
\usepackage{amsmath}
\makeatletter
\g@addto@macro\normalsize{%
\setlength\abovedisplayskip{-1pt}
\setlength\belowdisplayskip{0pt}
\setlength\abovedisplayshortskip{-1pt}
\setlength\belowdisplayshortskip{0pt}
}
\makeatother

\begin{document}

\title{Efficient and Effective Community Search on Large-scale Bipartite Graphs}

 \author{
 	\IEEEauthorblockN{Kai Wang$^{\dagger}$, Wenjie Zhang$^\dagger$, Xuemin Lin$^{\dagger}$, Ying Zhang$^\star$, Lu Qin$^\star$, Yuting Zhang$^{\dagger}$}
 	\IEEEauthorblockA{
 		$^\dagger$University of New South Wales,
         $^\star$University of Technology Sydney\\
 		kai.wang@unsw.edu.au, \{zhangw, lxue\}@cse.unsw.edu.au, \{ying.zhang, lu.qin\}@uts.edu.au,
 		ytzunsw@gmail.com
 	}
 }
\maketitle

\begin{abstract}

Bipartite graphs are widely used to model relationships between two types of entities. Community search retrieves densely connected subgraphs containing a query vertex, which has been extensively studied on unipartite graphs. However, community search on bipartite graphs remains largely unexplored. Moreover, all existing cohesive subgraph models on bipartite graphs can only be applied to measure the structure cohesiveness between two sets of vertices while overlooking the edge weight in forming the community. 
In this paper, we study the \wcom search problem on weighted bipartite graphs. Given a query vertex $q$, we aim to find the \wcom $\wcomg$ of $q$ which adopts \abcore to characterize the engagement level of vertices, and maximizes the minimum edge weight (significance) within $\wcomg$.

To support fast retrieval of $\wcomg$, we first retrieve the maximal connected subgraph of \abcore containing the query vertex (the \abcorecom), and the search space is limited to this subgraph with a much smaller size than the original graph. A novel index structure is presented which can be built in $O(\delta \cdot m)$ time and takes $O(\delta \cdot m)$ space where $m$ is the number of edges in $G$, $\delta$ is bounded by $\sqrt m$ and is much smaller in practice. Utilizing the index, the \abcorecom can be retrieved in optimal time. To further obtain $\wcomg$, we develop peeling and expansion algorithms to conduct searches by shrinking from the \abcorecom and expanding from the query vertex, respectively. The experimental results on real graphs not only demonstrate the effectiveness of the \wcom model but also validate the efficiency of our query processing and indexing techniques.
\end{abstract}

\section{Introduction}
\label{sct:introduction}

In many real-world applications, relationships between two different types of entities are modeled as bipartite graphs, such as customer-product networks \cite{wang2006unifying}, user-page networks \cite{beutel2013copycatch} and collaboration networks \cite{konect:DBLP}. Community structures naturally exist in these practical networks and {\em community search} has been extensively explored and proved useful on unipartite graphs \cite{CXWW14SIGMOD, SG10KDD, barbieri2015efficient, fang2016effective, huang2014querying, huang2015approximate, huang2017attribute, CSSurvey2020}. Given a query vertex $q$, {\em community search} aims to find communities (connected subgraphs) containing $q$ which satisfy specific cohesive constraints. {\color{black} In the literature, fair clustering methods \cite{chierichetti2017fair,ahmadi2020fair,ahmadian2020fair} are used to find communities (i.e., clusters) under fairness constraints on bipartite graphs. However, they aim to find a set of clusters under a global optimization goal and do not aim to search a personalized community for a specific user.} Nevertheless, no existing work has studied the {\em community search} problem on bipartite graphs. On bipartite graphs, various dense subgraph models are designed (e.g., ($\alpha, \beta$)-core \cite{liu2019,ding2017efficient}, bitruss \cite{wang2020efficient,zou2016bitruss, sariyuce2018peeling} and biclique \cite{zhang2014finding}) which can be used as the cohesive measurement of a community.
However, simply applying these cohesive measurements only ensures the structure cohesiveness of communities but ignores another important characteristic, the weight (or significance) of interactions between the two sets of vertices. For example, ($\alpha, \beta$)-core is defined as the maximal subgraph where each vertex in upper layer has at least $\alpha$ neighbors and each vertex in lower layer has at least $\beta$ neighbors. In the customer-movie network shown in Figure \ref{fig:motivation}, each edge has a weight denoting the rating of a user to a movie. If the ($\alpha, \beta$)-core model is applied to search a community of ``Eric'', e.g., the maximal connected subgraph of (3, 2)-core containing ``Eric'', we will get the community formed by the four users and the five movies on the left side. Note that, this community includes ``Alien'' (not liked by ``Andy'' or ``Kane'') and ``Taylor'' (who has less interest in this genre of movies). 


\begin{figure}[t]
\centering
\includegraphics[width=0.34\textwidth]{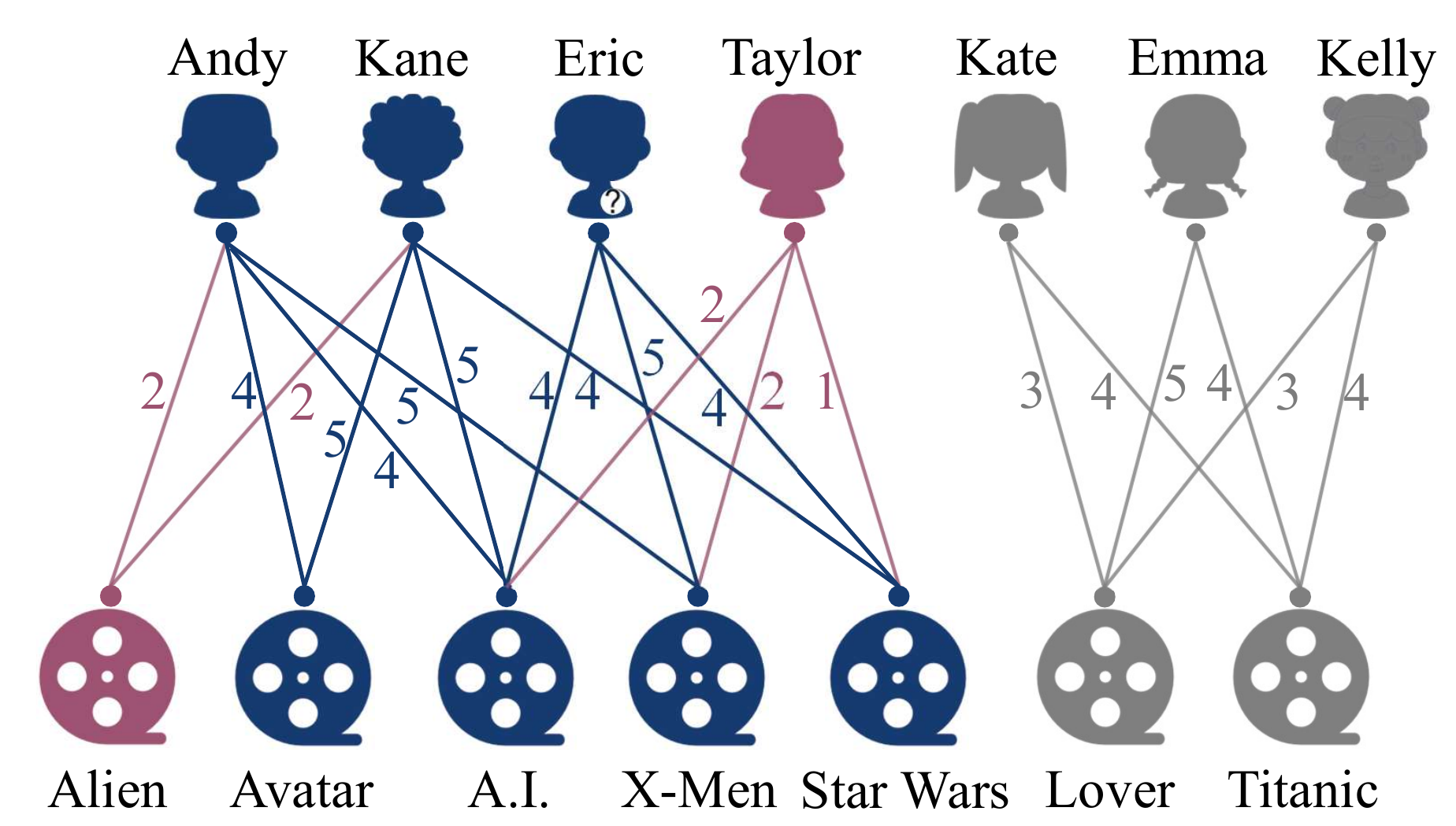}
\vspace{-1mm}
\caption{A user-movie network}
\label{fig:motivation}
\end{figure}

In this paper, we study the significant community search problem on weighted bipartite graphs, which is the first to study community search on bipartite graphs. Here, in a weighted bipartite graph $G$, each edge is associated with an edge weight. In addition, the weight (significance) of a community is measured by the minimum edge weight in it. {\color{black} A community with a high weight value indicates that every edge in the community represents a highly significant interaction.} We propose the \wcom model, which is the maximal connected subgraph containing the query vertex $q$ that satisfies the vertex degree constraint from \abcore, and has the highest graph significance. {\color{black} The intuition behind the new \wcom model is to capture structure cohesiveness as well as interactions (edges) with high significance. In addition, if we maximize the weight value under given $\alpha$ and $\beta$, we can find the most significant subgraph while preserving the structure cohesiveness.} For example, in Figure \ref{fig:motivation}, the subgraph in blue color, which excludes ``Alien'' and ``Taylor'', is the significant (3, 2)-community of ``Eric''. 

\noindent
{\bf Applications.} Finding the \wcom has many real-world applications and we list some of them below.

\noindent {\em $\bullet$ Personalized Recommendation.} In user-item networks, users leave reviews for items with ratings. Examples include viewer-movie network in \texttt{IMDB} ({https://www.imdb.com}), reader-book network in \texttt{goodreads} ({https://www.goodreads.com}), etc.  The platforms can utilize the \wcom model to provide personalized recommendations. For example, based on the community found in Figure \ref{fig:motivation}, we can put the people who give common high ratings (``Andy'' and ``Kane'') on the recommended friend list of the query user (``Eric''). We can also recommend the movie (``Avatar'') which the user is likely to be interested in to the query user (``Eric''). 

\noindent {\em $\bullet$ Fraud Detection.} In e-commerce platforms such as Amazon and Alibaba, customers and items form a customer-item bipartite graph in which an edge represents a customer purchased an item, and the edge weight measures the number of purchases or the total transaction amount. Fraudsters and the items they promote are prone to form cohesive subgraphs \cite{liu2019, wang2020efficient}. Since the cost of opening fake accounts is increased with the improvement of fraud detection techniques, frauds cannot rely on many fake accounts \cite{beutel2013copycatch}. Thus, the number of purchases or the total transaction amount per account is increased. Given a suspicious item or customer as the query vertex, our \wcom model allows us to find the most suspicious fraudsters and related items in the customer-item bipartite graphs and reduce false positives. 

\noindent {\em $\bullet$ Team Formation.} In a bipartite graph formed by developers and projects, an edge between a developer and a project indicates that the developer participates in the project, and the edge weight shows the corresponding contribution (e.g., number of tasks accomplished). A developer may wish to assemble a team with a proven track record of contributions in related projects, which can be supported by a \wcom search over the bipartite graph.

\noindent
{\bf Challenges.} {\color{black} To obtain the \wcom, we can iteratively remove the vertices without enough neighbors and the edges with small weights from the original graph. However, when the graph size is large and there are many vertices and edges that need to be removed, this approach is inefficient. For example, Figure \ref{fig:index_idea}(a) shows the graph $G$ with 2{,}003 edges. We need to remove 1{,}999 edges from $G$ to get the significant $(2, 2)$-community of $u_3$ with only 4 edges. 

In this paper, we focus on indexing-based approaches. A straightforward idea is precomputing all the \wcoms for all $\alpha$, $\beta$, and $q$ combinations. This idea is impractical since both structure cohesiveness and significance need to be considered. For different $q$ and $\alpha$, $\beta$ values, the \wcoms can be different and there does not exist hierarchical relationships among them. Therefore, we resort to a two-step approach. In the first step, we observe that the \abcorecom always contains the \wcom for a query vertex $q$. Here, \abcorecom is the maximal connected subgraph containing $q$ in the \abcore (without considering the edge weights). For example, Figure \ref{fig:index_idea}(b) shows the $(2, 2)$-community of $u_3$ which contains the significant $(2, 2)$-community of $u_3$ and is much smaller than the original graph $G$. Therefore, we try to index all \abcorecoms and use the one containing $q$ as the starting point when querying w.r.t. $q$. In the second step, we compute the \wcom based on the \abcorecom obtained in the first step. To make our ideas practically applicable, we need to address the following challenges.
\begin{enumerate}
    \item How to build an index to cover all \abcorecoms.
    \item How to bound the index size and the indexing time.
    \item How to efficiently obtain the \wcom from the \abcorecom of a query vertex.
\end{enumerate}

\begin{figure}[t]
\begin{centering}
\includegraphics[trim=0 10 0 15,width=0.46\textwidth]{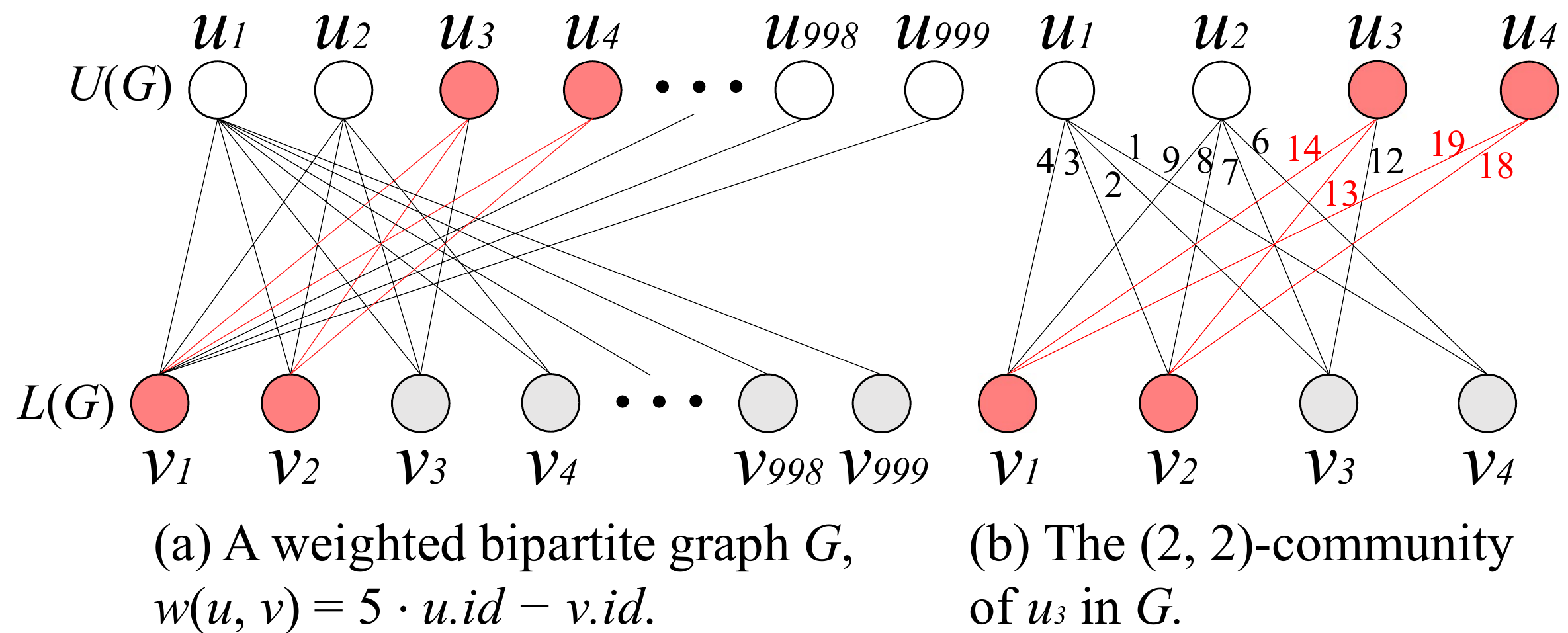}
\caption{An example graph, the significant $(2, 2)$-community of $u_3$ is marked in red color}
\label{fig:index_idea}
\vspace*{-1mm}
\end{centering}
\end{figure}

\noindent
{\bf Our approaches.} To address Challenge 1, we first propose the index $\indexbsa$ to store all the \abcorecoms. It is observed that the model of \abcore has a hierarchical property. In other words, $(\alpha,\beta)$-core $ \subseteq  (\alpha',\beta')$-core if $\alpha \geq \alpha'$ and $\beta \geq \beta'$. For example, in Figure \ref{fig:index_idea}, $G$ itself is the $(1, 1)$-core, the induced subgraph of $\{u_1, ..., u_{999}, v_1, v_2, v_3, v_4\}$ is the $(1, 2)$-core and we can obtain the $(1, 3)$-core from the $(1, 2)$-core by excluding $v_4$. Motivated by this observation, all the $(1,\beta)$-community with $\beta \geq 1$ can be organized hierarchically in the $(1, 1)$-core. For each vertex existing in $(1, 1)$-core, we sort its neighbors according to the maximal $\beta$ value where they exist in the $(1,\beta)$-core in non-increasing order. Then, when querying a $(1,\beta)$-community with $\beta \geq 1$, we only need to take the vertices and edges in this community using breath-first search. For example, if we want to query the $(1, 2)$-community of $u_1$, we first take the neighbors $\{v_1, v_2, v_3, v_4\}$ of $u_1$ and get $u_1$ to $u_{999}$ after searching from $v_1$. By organizing all the $(\alpha,1)$-cores where $\alpha \in [1, \alpha_{max}]$ in this manner, $\indexbsa$ can cover all the \abcorecoms. Similarly, we can also build the index $\indexbsb$ which stores all the $(1, \beta)$-core where $\beta \in [1, \beta_{max}]$ to cover all the \abcorecoms. Here $\alpha_{max}$ and $\beta_{max}$ are the maximal valid $\alpha$ and $\beta$ values in $G$ respectively. 


Reviewing $\indexbsa$ and $\indexbsb$, we observe that $\indexbsa$($\indexbsb$) can be very large when high degree vertices exist in $U(G)$($L(G)$). For example, $\indexbsa$ needs to store 999 copies of neighbors of $u_1$ since $u_1$ is contained in $(999, 1)$-core. The same issue occurs when $\indexbsb$ stores $v_1$'s neighbors. To handle this issue and address Challenge 2, we further propose the degeneracy-bounded index $\indexad$. Here, the degeneracy ($\delta$) is the largest number where the $(\delta, \delta)$-core is nonempty in $G$. Note that for each nonempty \abcore (or ($\alpha$, $\beta$)-community), we must have $min(\alpha, \beta) \leq \delta$. This is because it contradicts the definition of $\delta$ if an \abcore with $\alpha > \delta$ and $\beta > \delta$ exists. In addition, according to the hierarchical property of the \abcore model, all \abcorecoms with $\alpha \leq \beta$ can be organized in the $(\alpha, \alpha)$-core and all \abcorecoms with $\beta < \alpha$ can be organized in the $(\beta, \beta)$-core. In this manner, $\indexad$ only needs to store all the $(\tau, \tau)$-cores for each $\tau \in [1, \delta]$ to cover all the \abcorecoms. For example, in Figure \ref{fig:index_idea}, unlike $\indexbsa$ which needs to store $(1, 1)$-core to $(999, 1)$-core, $\indexad$ only needs to store $(1, 1)$-core, $(2, 2)$-core and $(3, 3)$-core since $\delta=3$. Since the size of each $(\tau, \tau)$-core ($\tau \in [1, \delta]$) is bounded by $O(m)$, $\indexad$ can be built in $O(\delta \cdot m)$ time and takes $O(\delta \cdot m)$ space to index all the \abcorecoms. 

To address Challenge 3, after retrieving the \abcorecom $\abcorecgq$, we first propose the peeling algorithm \peel which iteratively removes the edge with the minimal weight from $\abcorecgq$ to obtain $\wcomg$. For example, in Figure \ref{fig:index_idea}(b), to obtain the significant $(2, 2)$-community of $u_3$, the edge ($u_1$, $v_4$) is the first edge to be removed in \peel. Observing that $\wcomg$ can be much smaller than $\abcorecgq$ in many cases, we also propose the expansion algorithm \expand which iteratively adds the edge with maximal weights into an empty graph until $\wcomg$ is found. In \expand, we derive several rules to avoid excessively validating $\wcomg$. 
}

\noindent
{\bf Contribution.}
Our main contributions are listed as follows.

\begin{itemize}

\item We propose the model of \wcom which is the first to study community search problem on (weighted) bipartite graphs. 

\item  We develop a new two-step paradigm to search the \wcom. Under this two-step paradigm, novel indexing techniques are proposed to support the retrieval of the \abcorecom in optimal time. The index $\indexad$ can be built in $O(\delta \cdot m)$ time and takes $O(\delta \cdot m)$ space where $\delta$ is bounded by $\sqrt m$ and is much smaller in practice. Note that the proposed indexing techniques can also be directly applied to retrieve the \abcorecom on unweighted bipartite graphs in optimal time.

\item We propose efficient query algorithms to extract the \wcom from the \abcorecom.

\item We conduct comprehensive experiments on 11 real weighted bipartite graphs to evaluate the effectiveness of the proposed model and the efficiency of our algorithms.

\end{itemize}



\vspace{-0.2cm}
\section{Problem Definition}
\label{sct:preliminaries}






Our problem is defined over an undirected weighted bipartite graph $G(V\textnormal{=}(U, L), E)$, where $U(G)$ denotes the set of vertices in the upper layer, $L(G)$ denotes the set of vertices in the lower layer, $U(G) \cap L(G) = \emptyset$, $V(G) = U(G) \cup L(G)$ denotes the vertex set, $E(G) \subseteq U(G) \times L(G)$ denotes the edge set. An edge $e$ between two vertices $u$ and $v$ in $G$ is denoted as $(u, v)$ or $(v, u)$. The set of neighbors of a vertex $u$ in $G$ is denoted as $N(u, G) = \{ v\in V(G) \mid (u, v) \in E(G) \} $, and the degree of $u$ is denoted as $\degree(u, G) = |N(u, G)|$. We use $n$ and $m$ to denote the number of vertices and edges in $G$, respectively, and we assume each vertex has at least one incident edge. Each edge $e = (u, v)$ has a weight $w(e)$ (or $w(u, v)$). The size of $G$ is denoted as $\size(G) = |E(G)|$. 


\begin{definition}
\label{def:abcore}
{\bf (\abcore)} Given a bipartite graph G and degree constraints $\alpha$ and $\beta$, 
a subgraph $\abcoreg$ is the $(\alpha,\beta)$-core of $G$
if (1) $deg(u, \abcoreg) \geq \alpha$ for each $u \in U(\abcoreg)$ and 
$deg(v, \abcoreg) \geq \beta$ for each $v \in L(\abcoreg)$;
(2) $\abcoreg$ is maximal, i.e., any supergraph $G' \supset \abcoreg$ is not an \abcore. 
\end{definition}

\begin{definition}
\label{def:abcorec}
{\bf (($\alpha$, $\beta$)\textnormal{-}Connected Component)} Given a bipartite graph $G$ and its \abcore $\abcoreg$, 
a subgraph $\abcorecg$ is a \abcorec
if (1) $\abcorecg \subseteq \abcoreg$ and $\abcorecg$ is connected; (2) $\abcorecg$ is maximal, i.e., any supergraph $G' \supset \abcorecg$ is not a \abcorec. 
\end{definition}

\begin{definition}
\label{def:abcorecom}
{\bf (($\alpha$, $\beta$)\textnormal{-}Community)} Given a vertex $q$, we call the \abcorec containing $q$ the ($\alpha$, $\beta$)\textnormal{-}community, denoted as $\abcorecgq$.
\end{definition}
\begin{definition} {\bf (Bipartite Graph Weight)}
Given a bipartite graph $G$, the weight value of $G$ denoted by $f(G)$ is defined as the minimum edge weight in $G$. 
\end{definition}

After introducing the \abcore and bipartite graph weight, we define the \wcom as below.

\begin{definition} {\bf (Significant $(\alpha, \beta)$-Community)}
\label{def:wcom}
Given a weighted bipartite graph $G$, degree constraints $\alpha$, $\beta$ and query vertex $q$, a subgraph $\wcomg$ is the \wcom of $G$ if it satisfies the following constraints:
\begin{enumerate}
\item {\bf Connectivity Constraint.} $\wcomg$ is a connected subgraph which contains $q$; 
\item {\bf Cohesiveness Constraint.} Each vertex $u \in$ $U(\wcomg)$ satisfies $\degree(u, \wcomg) \geq \alpha$ and each vertex $v \in$ $L(\wcomg)$ satisfies $\degree(v, \wcomg) \geq \beta$;  
{\color{black}
\item {\bf Maximality Constraint.} There exists no other $G' \subseteq \abcorecgq$ satisfying constraints 1) and 2) with $f(G') > f(\wcomg)$. In addition, there exists no other supergraph $G'' \supset \wcomg$ satisfying constraints 1) and 2) with $f(G'') = f(\wcomg)$.
}
\end{enumerate}
\end{definition}

\noindent
\textbf{Problem Statement. }
Given a weighted bipartite graph $G$, parameters $\alpha$, $\beta$ and a query vertex $q$, the \textit{\wcom search} problem aims to find the \wcom (SC) in $G$.

{\color{black}
\begin{example}
Consider the bipartite graph $G$ in Figure \ref{fig:index_idea}(a). Figure \ref{fig:index_idea}(b) shows the (2,2)-community of $u_3$. In addition, the significant (2,2)-community of $u_3$ is shown in Figure \ref{fig:index_idea}(b) (in red color) which is formed by the edges $(u_3, v_1)$, $(u_3, v_2)$, $(u_4, v_1)$ and $(u_4, v_2)$.
\end{example}
}

\noindent
{\bf Solution Overview.} According to Definition \ref{def:abcorecom} and Definition \ref{def:wcom}, we have the following lemma.
{\color{black} 
\begin{lemma}
\label{lemma:wcom}
Given a weighted bipartite graph $G$, the significant $(\alpha, \beta)$-community is unique, which is a subgraph of the $(\alpha, \beta)$-community.
\end{lemma}

\vspace{-0.3cm}
\begin{proof}
Suppose there exist two different \wcoms $\wcomg_1$ and $\wcomg_2$ where $f(\wcomg_1)$ = $f(\wcomg_2)$, $\wcomg_1 \not \subseteq \wcomg_2$ and $\wcomg_2 \not \subseteq \wcomg_1$. Then $\wcomg_3 = \wcomg_1 \cup \wcomg_2$ satisfies constraints 1) and 2) in Definition \ref{def:wcom} with $f(\wcomg_3)$ = $f(\wcomg_1)$ = $f(\wcomg_2)$. This violates the maximality constraint in Definition \ref{def:wcom}. Thus, the \wcom is unique and is a subgraph of the $(\alpha, \beta)$-community by definition.
\end{proof}
}

Following the above lemma, we can use indexing techniques to efficiently find the \abcorecom first. In this manner, the search space is limited to a much smaller subgraph compared to $G$. Then, we further search on the \abcorecom to identify the \wcom. According to this two-step algorithmic framework, we present our techniques in the following sections.

\section{Retrieve the \abcorecom in optimal time}
\label{sct:index}
In this section, we explore indexing techniques to retrieve the \abcorecom in an efficient way. 



\subsection{Basic Indexes}
\label{sct:basic_index}
In \cite{liu2019}, the authors propose the \abindex which can obtain the vertex set of the \abcore (i.e., $V(\abcoreg)$) in optimal time. However, to obtain $\abcorecgq$ after having $V(\abcoreg)$, we still need to traverse all the neighbors of each vertex in $\abcorecgq$ (starting from the query vertex) including those neighbors which are not in $\abcorecgq$. This process needs $O(|V(\abcorecgq)| \cdot \sum_{v \in V(\abcorecgq)}{\degree(v, G)})$ time and when $\frac{|\size(\abcorecgq)|}{\sum_{v \in V(\abcorecgq)}{\degree(v, G)}}$ is small, it may need to access many additional edges not in the queried community. Motivated by this, we explore how to construct an index to support optimal retrieval of the \abcorecom (i.e., optimal retrieval of \abcorecs). 

By Definition \ref{def:abcore}, we have the following lemma.

\begin{lemma}
\label{lemma:abcore}
$(\alpha,\beta)$-core $ \subseteq  (\alpha',\beta')$-core if $\alpha \geq \alpha'$ and $\beta \geq \beta'$.
\end{lemma}

We also define the $\alpha$-offset and the $\beta$-offset of a vertex as follows. 

\begin{definition} {\bf ($\alpha$-/$\beta$-offset)}
\label{definition:offset}
Given a vertex $u \in V(G)$ and an $\alpha$ value, its $\alpha$-offset denoted as $s_a(u, \alpha)$ is the maximal $\beta$ value where $u$ can be contained in an \abcore. If $u$ is not contained in $(\alpha,1)$-core, $s_a(u, \alpha)$ = 0. Symmetrically, the $\beta$-offset $s_b(u, \beta)$ of $u$ is the maximal $\alpha$ value where $u$ can be contained in an \abcore.
\end{definition}

\begin{figure}[hbt]
\begin{centering}
\includegraphics[trim=0 10 0 15,width=0.46\textwidth]{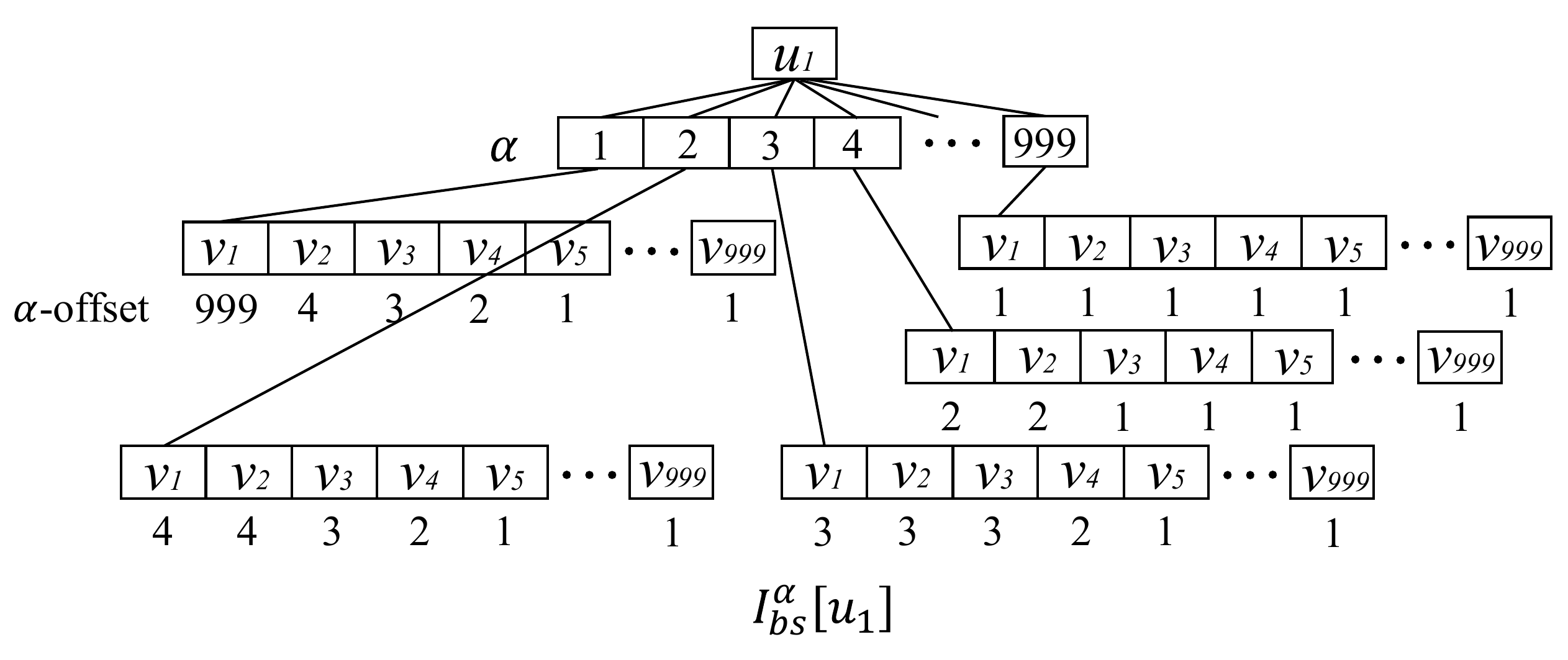}
\vspace{-2mm}
\caption{$\indexbsa[u_1]$ of $G$, edge weights are omitted}
\label{fig:index_bsa}
\end{centering}
\end{figure}

\begin{algorithm}[htbp]  
    \small
    \DontPrintSemicolon
    \caption{Index Construction of $\indexbsa$}  
    \label{algorithm:bsa_index_construction}
    \LinesNumbered
    \KwIn{$G$} 
    \KwOut{$\indexbsa$}
    $\alpha \gets 1$;\\
    {\color{black} $\alpha_{max} \gets$ the maximal vertex degree in $U(G)$;}\\
    \While{$\alpha \leq \alpha_{max}$} {
    compute $s_a(u, \alpha)$ for each vertex $u \in V(G)$;\\
    \ForEach{$u \in (\alpha, 1)\textnormal{-}core$}{
        \ForEach{$v \in N(u, G)$}{
            \If{$s_a(v, \alpha) \geq 1$}{
                $\indexbsa[u][\alpha] \gets \{v, w(u, v), s_a(v, \alpha$)\};\\
            }
        }
        sort $\indexbsa[u][\alpha]$ in decreasing order of their $\alpha$-offsets;\\
    }
    $\alpha \gets \alpha + 1$;\\
    }
    \Return{$\indexbsa$};
\end{algorithm}

\begin{figure}[t]
\begin{centering}
\includegraphics[trim=0 10 0 15,width=0.42\textwidth]{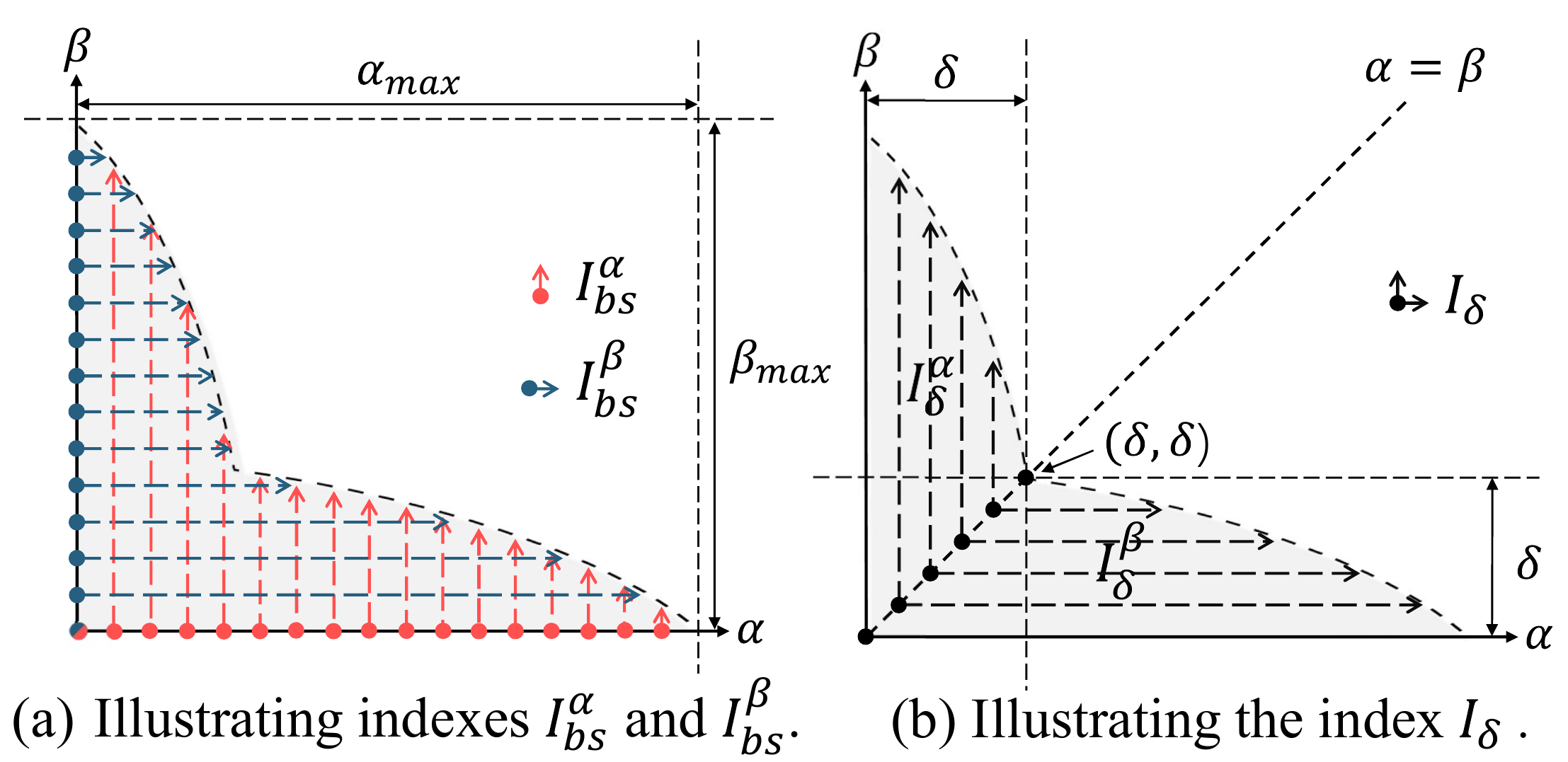}
\caption{Illustrating the ideas of indexing techniques}
\label{fig:index_g}
\vspace*{-1mm}
\end{centering}
\end{figure}

Since \abcore follows a hierarchical structure according to Lemma \ref{lemma:abcore}, an index can be constructed in the following way. For each vertex $u$, its $\alpha$-offset indicates that $u$ is contained in the $(\alpha,s_a(u, \alpha))$-core and is not contained in the $(\alpha, s_a(u, \alpha)\textnormal{+}1)$-core. According to Lemma \ref{lemma:abcore}, if $u$ is contained in the $(\alpha, s_a(u, \alpha))$-core, it is also contained in the \abcore with $\beta \leq s_a(u, \alpha)$. {\color{black} As shown in Figure \ref{fig:index_g}(a), the shaded area represents all the valid combinations of $\alpha$ and $\beta$ where an \abcorecom exists. As illustrated, we can organize the \abcores hierarchically and construct the basic index $\indexbsa$ as shown in Algorithm \ref{algorithm:bsa_index_construction}.  Firstly, we obtain $\alpha_{max}$ which is the maximal $\alpha$ value such that an $(\alpha, 1)$-core exists and it is equal to the maximal vertex degree in $U(G)$.} We then compute the $\alpha$-offset for each vertex. For each vertex $u$ and $\alpha$ combination (where $u$ exists in $(\alpha, 1)$\textnormal{-}core), we create an adjacent list $\indexbsa[u][\alpha]$ to store its neighbors. In $\indexbsa[u][\alpha]$, we sort $u$'s neighbors in non-increasing order of their $\alpha$-offsets and remove these neighbors with $\alpha$-offsets equal to zero. Figure \ref{fig:index_bsa} is an example which shows $\indexbsa[u_1]$ of $G$ in Figure \ref{fig:index_idea}(a). We can see that $\indexbsa[u_1]$ contains the neighbors of $u_1$ of different $\alpha$ values. 

\begin{algorithm}[htbp]  
    \small
    \DontPrintSemicolon
    \caption{Query based on $\indexbsa$}  
    \label{algorithm:query}
    \LinesNumbered
    \KwIn{$G$, $q$, $\alpha$, $\beta$, $\indexbsa$;} 
    \KwOut{$\abcorecgq$}
    $Q \gets q$;\\
    $visited(q) \gets true$;\\
    \While{$Q$ is not empty}{
        $u \gets Q.pop()$;\\
        \ForEach{$v \in \indexbsa[u][\alpha]$}{
            \If{$s_a{(v, \alpha)} \geq \beta$ }{
                $\abcorecgq \gets (u,v)$ if $u \in L(G)$;\\
                \If{$visited(v)=false$ }{
                    $Q.push(v)$;\\
                    $visited(v) \gets true$;\\
                }
            } \Else {
            \textbf{break};\\
            }
        }
    }
    \Return{$\abcorecgq$};
\end{algorithm}

\noindent
{\bf Optimal retrieval of $\abcorecgq$ based on $\indexbsa$.}
Given a query vertex $q$, Algorithm \ref{algorithm:query} illustrates the query process of the \abcorecom (i.e., $\abcorecgq$) based on $\indexbsa$. When querying $\abcorecgq$, we first put the query vertex into the queue. Then, we pop the vertex $u$ from the queue, and visit the adjacent list $\indexbsa[u][\alpha]$ to obtain the neighbors of $u$ with $\alpha$-offset $\geq \beta$. For each valid neighbor $v$, we add the edge $(u, v)$ into $\abcorecgq$ if $u \in L(G)$ to avoid duplication. Then, we put these valid neighbors into the queue and repeat this process until the queue is empty. Since the neighbors are sorted in non-increasing order of their $\alpha$-offsets, we can early terminate the traversal of the adjacent list when the $\alpha$-offset of a vertex is smaller than the given $\beta$. 

{\color{black}
\begin{lemma}
\label{lemma:optimal}
Given a bipartite graph $G$ and a query vertex $q$, Algorithm \ref{algorithm:query} computes $\abcorecgq$ in $O(\size(\abcorecgq))$ time, which is optimal.
\end{lemma}

\vspace{-0.4cm}
\begin{proof}
In Algorithm \ref{algorithm:query}, for each $u \in Q$, since there is no duplicate vertex in $\indexbsa[u][\alpha]$ and only its neighbor $v \in \indexbsa[u][\alpha]$ with $s_a{(v, \alpha)} \geq \beta$ can be accessed, each $u$ and $v$ combination corresponds to an edge in $\abcorecgq$. In addition, since each vertex can be only added once into $Q$ according to lines 8 - 10, Algorithm \ref{algorithm:query} computes $\abcorecgq$ in $O(\size(\abcorecgq))$ time, which is optimal as it is linear to the result size.
\end{proof}
}
\vspace{-0.3cm}
\begin{example}
Considering the graph in Figure \ref{fig:index_idea} and $\indexbsa[u_1]$ in Figure \ref{fig:index_bsa}, if we want to get the $(3, 3)$-community of $u_1$ $C_{3,3}{(u_1)}$, we first traverse $\indexbsa[u_1][3]$ to get all the neighbors with $\alpha$-offsets $\geq 3$ which are $v_1$, $v_2$ and $v_3$. The edges $(u_1, v_1)$, $(u_1, v_2)$ and $(u_1, v_3)$ will be added into $C_{3,3}{(u_1)}$. Then, we go to the index nodes $\indexbsa[v_1][3]$, $\indexbsa[v_2][3]$ and $\indexbsa[v_3][3]$ to get unvisited vertices $u_2$ and $u_3$ with $\alpha$-offsets $\geq 3$. The edges $(u_2, v_1)$, $(u_2, v_2)$, $(u_2, v_3)$, $(u_3, v_1)$, $(u_3, v_2)$, $(u_3, v_3)$ will be added into $C_{3,3}{(u_1)}$ when accessing $\indexbsa[u_2][3]$ and $\indexbsa[u_3][3]$.
\end{example}

In addition, apart from $\indexbsa$, we can construct an index $\indexbsb$ similarly based on $\beta$-offsets which also achieves optimal query processing. For each vertex $u$ and $\beta$ combination, we create an adjacent list to store its neighbors and we sort its neighbors in non-increasing order of their $\beta$-offsets (removing these neighbors with $\beta$-offsets = 0). When querying the $\abcorecgq$, we first go to the adjacent list indexing by $q$ and $\beta$, and obtain the neighbors of $q$ with $\beta$-offset $\geq \alpha$. Then we run a similar breadth-first search as Algorithm \ref{algorithm:query} shows. {\color{black} Using $\indexbsb$, we can also achieve optimal retrieval of $\abcorecgq$ which can be proved similarly as Lemma \ref{lemma:optimal}.}

\noindent
{\bf Complexity analysis of basic indexes.}
Storing $\indexbsa$ needs $\indexbsas$ = $O(\sum_{\alpha=1}^{\alpha_{max}}(\size((\alpha, 1)\textnormal{-}core))$ space. Since $\sum_{\alpha=1}^{\alpha_{max}}(\size((\alpha, 1)\textnormal{-}core) \leq \sum_{\alpha=1}^{\alpha_{max}}(\size((1, 1)\textnormal{-}core))$, $\indexbsas$ is also bounded by $O(\alpha_{max} \cdot m)$. Similarly, $\indexbsb$ needs $O(\sum_{\beta=1}^{\beta_{max}}(\size((1, \beta)\textnormal{-}core))$ = $O(\beta_{max} \cdot m)$ space. 

In addition, the time complexity of constructing $\indexbsa$ is $\indexbsat$ = $O(\alpha_{max} \cdot m)$. This is because for $\alpha$ from 1 to $\alpha_{max}$, we can perform the peeling algorithm on each $(\alpha, 1)$-core to get the $\alpha$-offset for each vertex first. This process needs $O(\alpha_{max} \cdot m)$ time. Then, for each vertex $u$, we create at most $\alpha_{max}$ adjacent lists to store its neighbors which needs $O(\alpha_{max} \cdot m)$ time. Similarly, the time complexity of constructing $\indexbsb$ is $\indexbsat$ = $O(\beta_{max} \cdot m)$.

\subsection{The Degeneracy-bounded Index $\indexad$}
\label{sct:delta_index}

Reviewing $\indexbsa$ and $\indexbsb$, we can see that it is hard to handle high degree vertices in $U(G)$($L(G)$) using $\indexbsa$($\indexbsb$). This is because if these vertices exist in an \abcore with large $\alpha$ (or $\beta$) value, according to Lemma \ref{lemma:abcore}, $\indexbsa$ or $\indexbsb$ may need large space to store several copies of the neighbors of these high degree vertices. For example, in Figure \ref{fig:index_bsa}, $\indexbsa$ needs to store multiple copies of neighbors of $u_1$ since $u_1$ is contained in $(999, 1)$-core. The same issue occurs when $\indexbsb$ stores $v_1$'s neighbors. Thus, in this part, we explore how to effectively handle these high degree vertices and build an index with smaller space consumption. 

Firstly, we give the definition of degeneracy as follows. 

\begin{definition} {\bf (Degeneracy)}
\label{definition:degeneracy}
Given a bipartite graph $G$, the degeneracy of $G$ denoted as $\delta$ is the largest number where $(\delta, \delta)$-core is nonempty in $G$.
\end{definition}

Note that, $\delta$ is bounded by $\sqrt{m}$ and in practice, it is much smaller than $\sqrt{m}$ \cite{liu2019}. 

\begin{lemma}
\label{lemma:degeneracy}
Given a bipartite graph $G$, a nonempty $(\alpha,\beta)$-core in $G$ must have $min(\alpha, \beta) \leq \delta$.
\end{lemma}

\vspace{-0.4cm}
{\color{black}
\begin{proof}
We prove this lemma by contradiction. Suppose a nonempty $(\alpha, \beta)$-core exists in $G$ with $\alpha < \beta$ and $\alpha > \delta$. Then we will have $\alpha \geq \delta+1$ and $\beta \geq \delta+1$ which contradicts to the definition of $\delta$. Similarly, we cannot have an nonempty $(\alpha, \beta)$-core existing in $G$ with $\beta < \alpha$ and $\beta > \delta$. Thus, a nonempty $(\alpha,\beta)$-core in $G$ must have $min(\alpha, \beta) \leq \delta$.
\end{proof}
}
\vspace{-0.2cm}

Based on Lemma \ref{lemma:degeneracy}, we can observe that, given query parameters $\alpha$ and $\beta$, a partial index of $\indexbsa$ which only stores adjacent lists of $u$ for each $u$ and $\alpha$ combinations with $\alpha \leq \delta$ is enough to handle queries when $\alpha = min(\alpha, \beta)$. Similarly, a partial index of $\indexbsb$ which only stores adjacent lists under ($u, \beta$) combinations with $\beta \leq \delta$ is enough to handle queries when $\beta = min(\alpha, \beta)$.
Based on the above observation, we propose the index $\indexad$ as follows.

\begin{figure}[hbt]
\begin{centering}
\includegraphics[trim=0 10 0 15,width=0.46\textwidth]{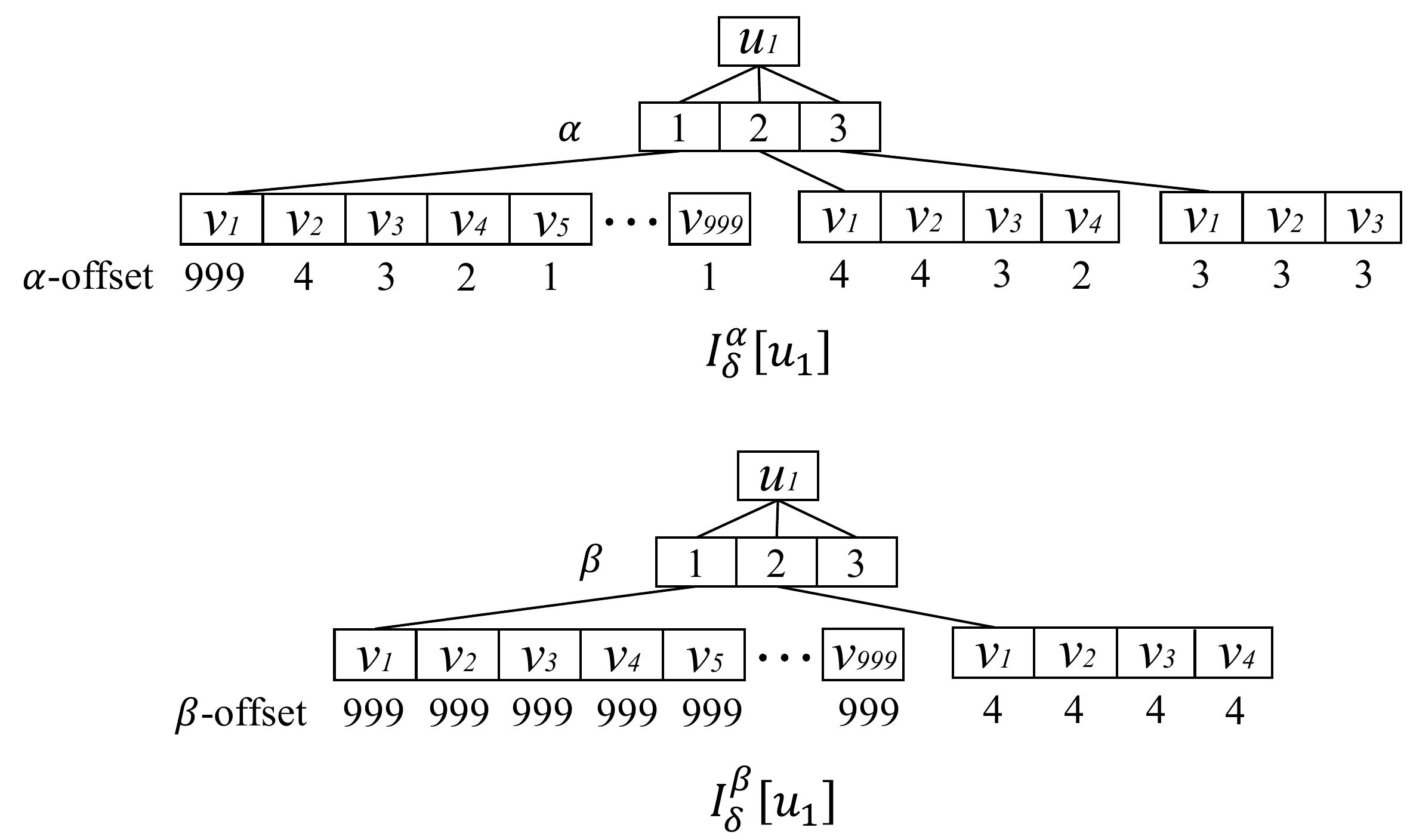}
\caption{$\indexad[u_1]$ of $G$, edge weights are omitted}
\label{fig:index_delta}
\end{centering}
\end{figure}

\noindent
{\bf Index Overview.} $\indexad$ contains two parts $\indexada$ and $\indexadb$ to cover all the $(\alpha, \beta)$-communities as illustrated in Figure \ref{fig:index_g}(b). 

In $\indexada$, for each vertex $u$ and $\alpha \leq \delta$ where $u$ exists in the $(\alpha, \alpha)$-core, we create an adjacent list $\indexada[u][\alpha]$ to store its neighbors. Note that, the neighbors are sorted in non-increasing order of their $\alpha$-offsets and the neighbors with $\alpha$-offsets less than $\alpha$ are removed. 

In $\indexadb$, for each vertex $u$ and $\beta \leq \delta$ where $u$ exists in the $(\beta,\beta)$-core, we create an adjacent list $\indexadb[u][\beta]$ to store its neighbors with $\beta$-offsets larger than $\beta$. The neighbors are sorted in non-increasing order of their $\beta$-offsets and the neighbors with $\beta$-offsets less or equal than $\beta$ are removed. Figure \ref{fig:index_delta} is an example of $\indexad[u_1]$ of $G$ in Figure \ref{fig:index_idea}(a). We can see that it consists of two parts $\indexada[u_1]$ and $\indexadb[u_1]$.

\noindent
{\bf Optimal retrieval of $\abcorecgq$ based on $\indexad$.} The query processing of $\abcorecgq$ based on $\indexad$ is similar to the query processing based on the basic indexes. The difference is that we need to choose to use $\indexada$ or $\indexadb$ at first. If the query parameter $\alpha \leq \delta$, we use $\indexada$ to support the query process. Otherwise, we go for $\indexadb$ to obtain the $\abcorecgq$. Since only valid edges are touched using $\indexad$, we can also obtain $\abcorecgq$ in $O(\size(\abcorecgq))$ time which is optimal. {\color{black} The proof of optimality is similar as Lemma \ref{lemma:optimal} and we omit it here due to the space limit.}

\begin{example}
Considering $G$ in Figure \ref{fig:index_idea} and $\indexad[u_1]$ in Figure \ref{fig:index_delta}, if we want to get the $(3, 3)$-community of $u_1$ $C_{3,3}{(u_1)}$, since $\alpha = \beta$, we first traverse $\indexada[u_1][3]$ to get all the neighbors with $\alpha$-offsets $\geq 3$, which are $v_1$, $v_2$ and $v_3$. The edges $(u_1, v_1)$, $(u_1, v_2)$ and $(u_1, v_3)$ will be added into $C_{3,3}{(u_1)}$. Then, we go to the index nodes $\indexada[v_1][3]$, $\indexada[v_2][3]$ and $\indexada[v_3][3]$ to get unvisited vertices $u_2$ and $u_3$ with $\alpha$-offsets $\geq 3$. The edges $(u_2, v_1)$, $(u_2, v_2)$, $(u_2, v_3)$, $(u_3, v_1)$, $(u_3, v_2)$, $(u_3, v_3)$ will be added into $C_{3,3}{(u_1)}$ when accessing $\indexada[u_2][3]$ and $\indexada[u_3][3]$.
\end{example}



\begin{lemma}
\label{lemma:indexad}
The space complexity of $\indexad$ denoted as $\indexads$ is $O(2 \cdot \sum_{\tau=1}^{\delta}\size(R_{\tau,\tau}))$ = $O(\delta \cdot m)$.
\end{lemma}
\vspace{-0.3cm}
{\color{black}
\begin{proof}
For each $\alpha \in [1, \delta]$ and $u \in R_{\alpha,\alpha}$, we need to store at most $\degree(u, R_{\alpha,\alpha})$ $u'$s neighbors in $\indexada$. Thus, $\indexada$ needs $O(\sum_{\alpha=1}^{\delta}\sum_{u \in R_{\alpha,\alpha}}\degree(u, R_{\alpha,\alpha}))$ = $O(\sum_{\alpha=1}^{\delta}\size(R_{\alpha,\alpha}))$=$O(\delta \cdot m)$ space. Similarly, $\indexadb$ also needs $O(\sum_{\beta=1}^{\delta}(\size(R_{\beta,\beta}))=O(\delta \cdot m)$ space. In total, the space for storing $\indexad$ is $O(\delta \cdot m)$.
\end{proof}
}
\vspace{-0.3cm}

\begin{algorithm}[htbp]  
    \small
    \DontPrintSemicolon
    \caption{Degeneracy-bounded Index Construction}  
    \label{algorithm:ad_index_construction}
    \LinesNumbered
    \KwIn{$G$} 
    \KwOut{$\indexad$}
    $\tau \gets 1$;\\
    {\color{black} compute $\delta$ using the $k$-core decomposition algorithm;}\\
    \While{$\tau \leq \delta$} {
    compute $\alpha$-offset $s_a(u, \tau)$ and $\beta$-offset $s_b(u, \tau)$ for each vertex $u \in V(G)$;\\
    \ForEach{$u \in (\tau, \tau)\textnormal{-}core$}{
        \ForEach{$v \in N(u, G)$}{
            \If{$s_a(v, \tau) \geq \tau$}{
                $\indexada[u][\tau] \gets \{v, w(u, v), s_a(v, \tau$)\};\\
            }
            \If{$s_b(v, \tau) > \tau$}{
               $\indexadb[u][\tau] \gets \{v,  w(u, v), s_b(v, \tau$)\};\\
            }
        }
        sort $\indexada[u][\tau]$ in decreasing order of their $\alpha$-offsets;\\
        sort $\indexadb[u][\tau]$ in decreasing order of their $\beta$-offsets;\\
    }
    $\tau \gets \tau + 1$;\\
    }
    \Return{$\indexad$};
\end{algorithm}

\noindent
{\bf Index Construction.} {\color{black} The construction algorithm of $\indexad$ is shown in Algorithm \ref{algorithm:ad_index_construction}. We first compute $\delta$ using the $k$-core decomposition algorithm in  \cite{khaouid2015k} since $\delta$ is equal to the maximum core number in $G$. Then, for each vertex $u$, we compute its $\alpha$-offset for each $\alpha \leq \delta$ and its $\beta$-offset for each $\beta \leq \delta$. These values can be obtained by the peeling algorithm in \cite{ding2017efficient}. Then, we loop $\tau$ from 1 to $\delta$ and add the valid neighbors of the vertices in the $(\tau, \tau)\textnormal{-}core$ into $\indexad$.

\begin{lemma}
\label{lemma:indexadc}
The time complexity of Algorithm \ref{algorithm:ad_index_construction} is $O(\delta \cdot m)$.
\end{lemma}
\vspace{-0.2cm}
\begin{proof}
\label{proof:indexadc}
For each $\tau$, we can first obtain the $(\tau, 1)$-core and the $\alpha$-offsets of all the vertices can be computed using the core decomposition algorithm \cite{khaouid2015k} in $O(m)$ time. The $\beta$-offsets of all the vertices can also be computed in $O(m)$ time similarly. Then, sorting $\indexada[u][\tau]$ and $\indexadb[u][\tau]$ for each vertex $u$ also needs $O(m)$ time in total by using bin sort \cite{khaouid2015k}. Since $\tau \in [1, \delta]$, the time complexity of Algorithm \ref{algorithm:ad_index_construction} is $O(\delta \cdot m)$.
\end{proof}
}

{\color{black}
\noindent
{\bf Discussion of index maintenance.} When graphs are updated dynamically, it is inefficient to reconstruct the indexes from scratch. Thus, we discuss the main idea of the incremental algorithms for maintaining $\indexad$. Other indexes in this paper can be maintained in a similar way. 

\noindent
{\em \underline{Edge insertion.}} Suppose an edge $(u, v)$ is inserted into $G$. For each $\alpha \leq \delta$, we first add $u(v)$ into $\indexada[v][\alpha]$ ($\indexada[u][\alpha]$) if $s_a(u, \alpha) \geq \alpha$ ($s_a(v, \alpha) \geq \alpha$). Then, for each $\alpha \leq \delta$, we track changes of the $\alpha$-offsets of the vertices. Note that, only the $\alpha$-offsets of the vertices in $S^+_{\alpha} = V(C_{\alpha,s_a(u, \alpha)}(u)) \cup V(C_{\alpha,s_a(v, \alpha)}(v))$ can be changed. This is because for each vertex not in $S^+_{\alpha}$, it either does not connect to $u$($v$) or $u$($v$) already exists in any $(\alpha, \beta)$-connected component it belongs to when fixing $\alpha$. Thus, we obtain the induced subgraph of $S^+_{\alpha}$ from $\indexad$ and compute the new $\alpha$-offsets of the vertices in $S^+_{\alpha}$ by peeling the subgraph. If the $\alpha$-offset of the vertex $u' \in S^+_{\alpha}$ is changed, we only need to update $\indexada[v'][\alpha]$ where $v' \in N(u', G)$. Similarly, for each $\beta \leq \delta$, only the $\beta$-offsets of the vertices in $S^+_{\beta} = V(C_{s_b(u, \beta), \beta}(u)) \cup V(C_{s_b(v, \beta), \beta}(v))$ can be changed. We compute the new $\beta$-offsets of these vertices and update $\indexadb$ in a similar way. Note that after the new edge is inserted, the value of $\delta$ can be increased by 1. If $\delta$ is increased, we compute the new index elements for $\delta+1$. 

\noindent
{\em \underline{Edge removal.}} Suppose an edge $(u, v)$ is removed from $G$. For each $\alpha \leq \delta$, we first remove $u(v)$ from $\indexada[v][\alpha]$ ($\indexada[u][\alpha]$) if $s_a(u, \alpha) \geq \alpha$ ($s_a(v, \alpha) \geq \alpha$). Similar as the insertion case, for each $\alpha$, only the $\alpha$-offsets of the vertices in $S^-_{\alpha} = V(C_{\alpha,1}(u)$\textbackslash $C_{\alpha,s_a(u, \alpha)+1}(u)) \cup V(C_{\alpha,1}(v)$\textbackslash $C_{\alpha,s_a(v, \alpha)+1}(v))$ can be changed. Thus, we recompute the $\alpha$-offsets of these vertices and update $\indexada$. $\indexadb$ can also be updated similarly.

\noindent
{\bf Remark.} Although we are dealing with the weighted bipartite graph in this work, the indexing techniques proposed in this section can directly support finding the \abcorecom on unweighted bipartite graph. 
}
\section{Query the \wcom}
\label{sct:query}
According to the definition of \wcom, the subgraph $\abcorecgq$ obtained from the index already satisfies the connectivity constraint and the cohesiveness constraint. Thus, in this section, we introduce two query algorithms to obtain the \wcom from $\abcorecgq$ {\color{black} to further satisfy} the maximality constraint. 

\subsection{Peeling Approach}

\begin{algorithm}[htbp]  
    \small
    \DontPrintSemicolon
    \caption{\peel}  
    \label{algorithm:peel}
    \LinesNumbered
    \KwIn{$G$, $q$, $\alpha$, $\beta$;} 
    \KwOut{$\wcomg$} 
    get $\abcorecgq$ from the index;\\
    $S \gets \emptyset$; $Q \gets \emptyset$;\\
    {\color{black} sort edges of $\abcorecgq$ in non-decreasing order by weights;}\\
    \While{$\abcorecgq$ is not empty}{
        {\color{black}$\kwnospace{w_{min}} \gets$ the minimal edge weight in $\abcorecgq$\\
        \ForEach{$(u,v) \in \abcorecgq$ with $w(u,v) = \kwnospace{w_{min}}$}{
            remove $(u,v)$ from $\abcorecgq$;\\
            $S.add((u,v))$;\\
            \If{$deg(u, \abcorecgq) < \alpha \wedge u \notin Q$}{
                $Q.push(u)$;
            }
            \If{$deg(v, \abcorecgq) < \beta \wedge v \notin Q$}{
                $Q.push(v)$;
            }
        }}
        \While{$Q$ is not empty}{
            $u' \gets Q.pop()$;\\
            \ForEach{$v' \in N(u', \abcorecgq)$}{
                remove $(u',v')$ from $\abcorecgq$;\\
                $S.add((u',v'))$;\\
                \If{$v'$ does not have enough degree}{
                    $Q.push(v')$;\\
                    \If{$v'$=$q$}{
                        $G' \gets S \cup \abcorecgq$;\\
                        Obtain $\wcomg$ from $G'$\;
                        \Return{$\wcomg$};\\
                    }
                }
            }
        }
        $S = \emptyset$;\\
    }
\end{algorithm}

Here, we introduce the peeling approach as shown in Algorithm \ref{algorithm:peel}. Firstly, we retrieve $\abcorecgq$ based on the indexes proposed in Section \ref{sct:index}. {\color{black} Note that if all the edge weights are equal in $\abcorecgq$, we can just return $\abcorecgq$ as the result. Otherwise, we sort the edges in $\abcorecgq$ in non-decreasing order by weights and we initialize an edge set $S$ and a queue $Q$ to empty. After that, we run the peeling process on $\abcorecgq$. In each iteration, we remove each edge $(u,v)$ with the minimal weight in $\abcorecgq$.} Also, we add $(u,v)$ into an edge set $S$ which records the edges removed in this iteration. Due to the removal of $(u,v)$, there may exist many vertices which do not have enough degree to stay in $\abcorecgq$ (i.e., for vertex $u \in U(\abcorecgq)$, $deg(u, \abcorecgq) < \alpha$ or for vertex $v \in L(\abcorecgq)$, $deg(v, \abcorecgq) < \beta$), we also remove the edges of these vertices and add the edges into $S$. We run the peeling process until $q$ does not satisfy the degree constraint. Then, we create $G'$= $S \cup \abcorecgq$ since the edges removed in this iteration need to be recovered to form the $\wcomg$. {\color{black} Finally, we remove the vertices without enough degree in $G'$ and run a breath-first search from $q$ on $G'$ to get the connected subgraph containing $q$ which is $\wcomg$.}




\noindent
\begin{theorem}
\label{theorem:peelc}
The \peel algorithm correctly solves the \wcom search problem.
\end{theorem}

\begin{proof}
\label{proof:newc}
According to Lemma \ref{lemma:wcom}, $\wcomg$ is a subgraph of $\abcorecgq$. Suppose there is a $G' \subseteq \abcorecgq$ satisfying the connected constraint and the cohesiveness constraint and has $f(G') > f(\wcomg)$. Since we always peel the edge with the minimal weight, $G'$ will be found after $\wcomg$. Since we peel $\abcorecgq$ until the degree of $q$ is not enough, $q \in G'$ will not have enough degree which contradicts the cohesiveness constraint. For the same reason, there exists no $G'' \supset \wcomg$ with $f(G'') = f(\wcomg)$. Thus, this theorem holds.
\end{proof}


\noindent
{\bf Time complexity.} \peel has three phases. Retrieving $\abcorecgq$ based on the index needs $(\size(\abcorecgq))$ time. Then, sorting the edges in $\abcorecgq$ needs $\sort(\abcorecgq)$ time which will be $O(\size(\abcorecgq)\cdot(log(\size(\abcorecgq))))$ if we use quick sort or $O(m')$ if we use bin sort where $m'$ equals to the maximal weight in $\abcorecgq$. After that, the whole peeling process requires $O(\size(\abcorecgq)$ time. In total, the time complexity of \peel is $O(\sort(\abcorecgq) + \size(\abcorecgq))$.

\noindent
{\bf Space complexity.} In the \peel algorithm, we need only $O(\size(\abcorecgq))$ space to store the edges in $\abcorecgq$ apart from the space used by the indexes.

\subsection{Expansion Approach}

Unlike the peeling approach which iteratively removes the edge with the minimal weight from $\abcorecgq$, in this part, we introduce the expansion approach \expand. \expand first initializes a subgraph $G^*$ as empty. Then it iteratively adds the edges with the maximal weight to $G^*$ (from $\abcorecgq$) until $G^*$ contains $\wcomg$. In this manner, if $\size(\wcomg)$ is much smaller than $\size(\abcorecgq)$, \expand can retrieve $\wcomg$ in a more efficient way compared to the peeling approach.

{\color{black} Following the above idea, we add edges with the maximal weight in $\abcorecgq$ to $G^*$ (and remove them from $\abcorecgq$) in each iteration. However, when adding an edge into $G^*$, it may not connect to $q$. Note that, we cannot discard these edges immediately since they may be connected to $q$ due to the later coming edges. Thus, the connected subgraphs in $G^*$ should be maintained in each iteration. With the help of union-find data structure  \cite{cormen2009introduction}, the connected subgraphs in $G^*$ can be maintained in constant amortized time, and we can efficiently obtain the connected subgraph containing $q$ in $G^*$.} 



\noindent
{\bf Checking the existence of $\wcomg$ in $C^*$.} Suppose $C^*$ is the connected subgraph containing $q$ in $G^*$, we can easily observe that $\wcomg$ can only be found in the iteration where $C^*$ is changed. In addition, we have the following bounds which can let us know whether $\wcomg$ is contained in $C^*$.

\begin{lemma}
\label{lemma:e1}
Given a connected subgraph $C^*$, if $\wcomg \subseteq C^*$, we have:
\[\alpha  \beta-\alpha-\beta \le |E(C^*)|-|U(C^*)|-|L(C^*)|\]
\end{lemma}

\begin{proof}
Since $C^*$ is a connected subgraph, we have $|E(C^*)| \ge |U(C^*)|+|L(C^*)|-1$.
According to the cohesiveness constraint of $\wcomg$, $\wcomg$ has at least $max\{\alpha\cdot|U(\wcomg)|, \beta\cdot|L(\wcomg)|\}$ edges. In addition, the number of incident edges of vertices in $V(C^*) \setminus V(\wcomg)$ is at least $|U(C^*)|+|L(C^*)|-|U(\wcomg)|-|L(\wcomg)|$ to ensure $C^*$ is connected. 

Hence, when $\alpha\cdot|U(\wcomg)|\ge \beta\cdot|L(\wcomg)|$, $|E(C^*)|\ge |U(C^*)|+|L(C^*)|-|U(\wcomg)|-|L(\wcomg)|+\alpha\cdot|U(\wcomg)|$.
It is immediate that $\alpha\le|L(\wcomg)|$ and $\beta\le|U(\wcomg)|$. Thus, we have $(\alpha-1)\cdot|U(\wcomg)|-|L(\wcomg)| \le |E(C^*)|-|U(C^*)|-|L(C^*)|$. By transformation, we have  $(\alpha-1)\cdot\beta-\alpha\le |E(C^*)|-|U(C^*)|-|L(C^*)|$. Then, we get $\alpha\beta-\alpha-\beta\le |E(C^*)|-|U(C^*)|-|L(C^*)|$.

When $\alpha\cdot|U(\wcomg)|< \beta\cdot|L(\wcomg)|$, $|E(C^*)| \ge |U(C^*)|+|L(C^*)|-|U(\wcomg)|-|L(\wcomg)|+\beta\cdot|L(\wcomg)|$, we can also get $\alpha\beta-\alpha-\beta\le |E(C^*)|-|U(C^*)|-|L(C^*)|$.
\end{proof}

\begin{lemma}
\label{lemma:e2}
Given a connected subgraph $C^* \subseteq G$, if $\wcomg \subseteq C^*$, it must contain $\alpha$ vertices where each vertex $u$ of them has $deg(u, C^*) \geq \beta$, and it must contain $\beta$ vertices where each vertex $v$ of them has $deg(v, C^*) \geq \alpha$. In addition, the query vertex should be one of these vertices.
\end{lemma}

\begin{proof}
\label{proof:e2}
This lemma directly follows from Definition \ref{def:wcom}.
\end{proof}

Based on the above lemmas, we can skip checking the existence of $\wcomg$ if the constraints are not satisfied. It is still costly if we check each $C^*$ satisfies the constraints since we need to perform the peeling algorithm on $C^*$ using $O(\size(C^*))$ time. To mitigate this issue, we set an expansion parameter $\epsilon > 1$ to control the number of checks. Firstly, we check $C^*$ when it first satisfies the constraints in the Lemma \ref{lemma:e1} and Lemma \ref{lemma:e2}. After that, we only check $C^*$ if its size is at least $\epsilon$ times than the size of its last check. Here we choose $\epsilon = 2$ and the reasons are as follows. Suppose for each $C^*_i$ ($i \in [1, d]$, $d$ is the total number of checks) which needs to be checked, $\size(C^*_i)$ is exactly $\epsilon$ times of $\size(C^*_{i-1})$. Since we can find $\wcomg$ in the final check, we have $\size(C^*_d) < \epsilon(\size(\wcomg))$. The time complexity of using the peeling algorithm to check all these connected subgraphs is $O(\sigma_{i=1}^{d}\size(C^*_i))$, and $\sigma_{i=1}^{d}\size(C^*_i) $ = $\size(C^d)+\frac{1}{\epsilon}\size(C^d)+\frac{1}{\epsilon^2}\size(C^d)+...+\frac{1}{\epsilon^{d}}\size(C^d)$, we can have $O(\sigma_{i=1}^{d}\size(C^*_i)) = O(\epsilon(\frac{1}{\epsilon-1})\size(\wcomg))$. We choose $\epsilon = 2$ since $\frac{1}{\epsilon-1}$ achieves the smallest value at $\epsilon = 2$.




\noindent
{\bf The \expand Algorithm.} We present the \expand algorithm as shown in Algorithm \ref{algorithm:expand}.  Firstly, we retrieve $\abcorecgq$ based on the indexes proposed in Section \ref{sct:index}. {\color{black} We can return $\abcorecgq$ if all the edge weights are equal in $\abcorecgq$. Otherwise, we sort the edges in $\abcorecgq$ in non-increasing order by weights and we initialize $G^*$ and $C^*$ to empty. After that, we iteratively add each edge $(u,v)$ with the maximal weight (in $\abcorecgq$) to $G^*$ and remove the added edge from $\abcorecgq$.} Note that the size and edges of the connected subgraphs in $G^*$ will be maintained using the union-find structure. If $C^*$ is changed, we will check whether $C^*$ satisfies the constraints in the Lemma \ref{lemma:e1} and Lemma \ref{lemma:e2}. After that, we will check if its size grows at least $\epsilon$ times. If it is, we run the peeling process to check whether $\wcomg$ is contained by $C^*$. {\color{black}In this peeling process, we iteratively remove all the vertices without enough degree from $C^*$. If $q$ is not removed from $C^*$, we run Algorithm \ref{algorithm:peel} to obtain $\wcomg$.} The algorithm finishes if it finds $\wcomg$ in $C^*$.

\begin{algorithm}[htbp]  
    \small
    \DontPrintSemicolon
    \caption{\expand}  
    \label{algorithm:expand}
    \LinesNumbered
    \KwIn{$G$, $q$, $\alpha$, $\beta$, $\epsilon$;} 
    \KwOut{$\wcomg$} 
    $G^* \gets \emptyset$; $C^* \gets \emptyset$; $\kwnospace{pre\_size} = 0$;\\
    get $\abcorecgq$ from the index;\\
    {\color{black} sort edges of $\abcorecgq$ in non-increasing order by weights;}\\
    \While{$\abcorecgq$ is not empty}{
        {\color{black}$\kwnospace{w_{max}} \gets$ the maximal edge weight in $\abcorecgq$\\
        \ForEach{$(u,v) \in \abcorecgq$ with $w(u,v) = \kwnospace{w_{max}}$}{        
        remove $(u,v)$ from $\abcorecgq$;\\
        $G^*.add((u,v))$;\\
        maintain the connected subgraphs in $G^*$;\\
        }}
        \If{$C^*$ is not changed or violates constraints in Lemma \ref{lemma:e1} and Lemma \ref{lemma:e2}}{
            {\bf continue};\\
        }        
        \If{$\size(C^*) \geq \kwnospace{pre\_size} \cdot \epsilon$}{
            $\kwnospace{pre\_size} \gets \size(C^*)$;
        } \Else {
            {\bf continue};\\
        }
        {\color{black}Remove the vertices without enough degree from $C^*$;\\
        \If{$q \in C^*$}{
        run Algorithm \ref{algorithm:peel} lines 3 - 23, replace $\abcorecgq$ with a copy of $C^*$}}
    }
\end{algorithm}



\noindent
\begin{theorem}
\label{theorem:newco}
The \expand algorithm correctly solves the \wcom search problem.
\end{theorem}

\begin{proof}
\label{proof:expandc}
According to Definition \ref{def:wcom}, $\wcomg$ is a subgraph of $\abcorecgq$. Since we always expand the edge with the maximal weight, the connected subgraph $C^*$ will always contain all the edges in $\abcorecgq$ which is connected to $q$ with weights $\geq f(C^*)$. According to Theorem \ref{theorem:peelc}, \peel can correctly check whether $\wcomg$ exists in $C^*$. Thus, this theorem holds.
\end{proof}

\noindent
{\bf Time complexity.} In \expand, retrieving $\abcorecgq$ based on the index needs $O(\size(\abcorecgq))$ time. Then, sorting the edges in $\abcorecgq$ needs $O(\sort(\abcorecgq))$ time. After that, the whole expansion process requires $O(\sum_{i=1}^{d}\size(C^*_i))$ time where $d$ is the number of subgraphs which survive to Algorithm \ref{algorithm:expand} line 16. In total, the time complexity of \expand is $O(\sort(\abcorecgq) + \sigma_{i=1}^{d}\size(C^*_i))$.

\noindent
{\bf Space complexity.} In the \expand algorithm, we need $O(\size(\abcorecgq))$ space to store the edges in $\abcorecgq$ except the space used by indexes.

{\color{black}
\noindent
{\bf Remark.} One may also consider using binary search over the weights to find $\wcomg$. To validate each weight, it still needs to run the peeling process which needs $O(m)$ time. In addition, if the result is found under a weight threshold, the algorithm stops and the search space does not need to be reduced anymore. Thus, this binary search method only needs to expand the search space which is similar to \expand. We implement the binary search approach and find its running time is similar to that of \expand (0.86$\times$-1.08$\times$) on all the datasets. Note that when the number of distinct weight values are small, \binary can have better performance than \expand.

}

\begin{table}[ht]
\small
\caption{\label{table:datasets}
Summary of Datasets}
\vspace*{-4mm}
\begin{center}
\scalebox{0.80}{
\begin{tabular}{c||c|c|c|c|c|c|c}
\noalign{\hrule height 1pt}
\textbf{Dataset} &$|E|$  & $|U|$ & $|L|$  & $\delta$ & $\alpha_{max}$ & $\beta_{max}$ & ${\color{black} |R_{\delta,\delta}|}$\\\hline
 BS & 433K & 77.8K & 186K & 13 & 8{,}524 & 707 &{\color{black} 13.6K}\\ \hline
 GH & 440K & 56.5K  & 121K & 39 & 884 &3{,}675&{\color{black} 21.5K}\\ \hline
 SO & 1.30M & 545K & 96.6K & 22 &4{,}917 & 6{,}119&{\color{black} 13.0K}\\ \hline
 LS & 4.41M & 992 & 1.08M & 164 & 55{,}559 & 773&{\color{black} 177K}\\ \hline
 DT& 5.74M & 1.62M & 383 & 73 & 378& 160{,}047&{\color{black} 30.5K}\\ \hline
 AR & 5.74M & 2.15M & 1.23M & 26 & 12{,}180 & 3{,}096&{\color{black} 36.6K}\\ \hline
 PA & 8.65M & 1.43M & 4.00M & 10 & 951 & 119&{\color{black} 639}\\ \hline
 ML & 25.0M & 162K & 59.0K & 636 & 32{,}202 & 81{,}491&{\color{black} 2.12M}\\ \hline
 DUI & 102M & 833K & 33.8M & 183 & 24{,}152 & 29{,}240&{\color{black} 2.30M}\\ \hline
 EN & 122M  & 3.82M & 21.5M  & 254& 1{,}916{,}898 & 62{,}330&{\color{black} 1.03M}\\ \hline
DTI& 137M & 4.51M & 33.8M & 180& 1{,}057{,}753 & 6{,}382&{\color{black} 242K}\\
\noalign{\hrule height 1pt}
\end{tabular}}
\end{center}
\end{table}

\section{Experiments}
\label{sct:experiment}
In this section, we first evaluate the effectiveness of the \wcom model. Then, we evaluate the efficiency of the techniques for retrieving \abcorecoms and \wcoms. 

\subsection{Experiments setting}
\noindent
{\bf Algorithms.}
Our empirical studies are conducted against the following designs:

\noindent
{\em $\bullet$ Techniques to retrieve the \abcorecom.}
The query algorithms: 1) the online query algorithm $Q_o$ in \cite{ding2017efficient}, and the query algorithms based on the following indexes: 2) $Q_v$ based on the \abindex $\indexbc$ proposed in \cite{liu2019}, 3) $Q_{opt}$ based on the degeneracy-bounded index $\indexad$ in Section \ref{sct:delta_index}. The indexes: 1) the \abindex $\indexbc$, 2) basic indexes $\indexbsa$ and $\indexbsb$, 3) $\indexad$.

\noindent
{\em $\bullet$ Algorithms to retrieve the \wcom.} 1) the peeling algorithm \peel, 2) the expansion algorithm \expand in Section \ref{sct:query} and 3) {\color{black} a baseline algorithm \baseline which iteratively expands the edges (with larger weight value) from the connected component containing $q$ of the whole graph rather than from $\abcorecgq$.}

The algorithms are implemented in C++ and the experiments are run on a Linux server with Intel Xeon 2650 v3 2.3GHz processor and 768GB main memory. {\em We terminate an algorithm if the running time is more than $10^4$ seconds}.

\noindent
{\bf Datasets.}
We use $11$ real datasets in our experiments which are \texttt{Bookcrossing} (BC), \texttt{Github} (GH), \texttt{StackOverflow} (SO), \texttt{Lastfm} (LS), \texttt{Discogs} (DT), \texttt{Amazon} (AR), \texttt{DBLP} (PA), \texttt{MovieLens} (ML), \texttt{Delicious-ui} (DUI), \texttt{Wikipedia-en} (EN) and \texttt{Delicious-ti} (DTI). All the datasets we use can be found in KONECT ({http://konect.uni-koblenz.de}). {\color{black}Note that, for the datasets without weights (i.e., DT and PA), we choose the random walk with restart model \cite{tong2006fast} to compute the node relevance and generate the weights. Here several other models \cite{chen2020structsim,jeh2002simrank} can also be applied.} 

The summary of datasets is shown in Table \ref{table:datasets}. $U$ and $L$ are vertex layers, $|E|$ is the number of edges. $\delta$ is the degeneracy. $\alpha_{max}$ and $\beta_{max}$ are the largest value of $\alpha$ and $\beta$ where a $(\alpha,1)$-core or $(1,\beta)$-core exists, respectively. {\color{black}$ |R_{\delta,\delta}|$ denotes the number of edges in $R_{\delta,\delta}$ in each dataset.} In addition, $M$ denotes $10^6$ and $K$ denotes $10^3$. 


\subsection{Effectiveness evaluation}

{\color{black}
In this section, we evaluate the effectiveness of our model on \texttt{MovieLense} which contains 25M ratings (ranging from 1 to 5) from 162K users ($U$) on 59K movies ($L$). 

We compare the \wcom model with the $(\alpha, \beta)$-core, $k$-bitruss (setting $k=\alpha \cdot \beta$) \cite{zou2016bitruss} and maximal biclique \cite{zhang2014finding} models. We also add a community $C_{4\star}$ which is the induced subgraph of all the movies with average ratings at least 4. Note that, we use the connected components of the query vertex as the result when considering different models. 

\vspace*{-2mm}
\begin{figure}[htb]
\begin{centering}
\includegraphics[trim=-130 0 0 0,clip,width=0.43\textwidth]{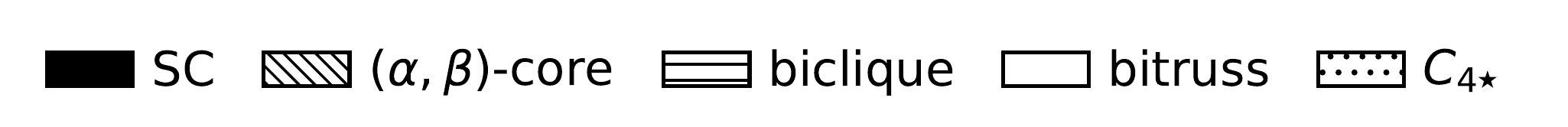}
\vspace*{-5mm}
\end{centering}
\end{figure}

\begin{figure}[htb]
\begin{centering}
\subfigure[Bipartite graph density]{
\begin{minipage}[b]{0.22\textwidth}
\includegraphics[trim=0 0 0 0,clip,width=1\textwidth]{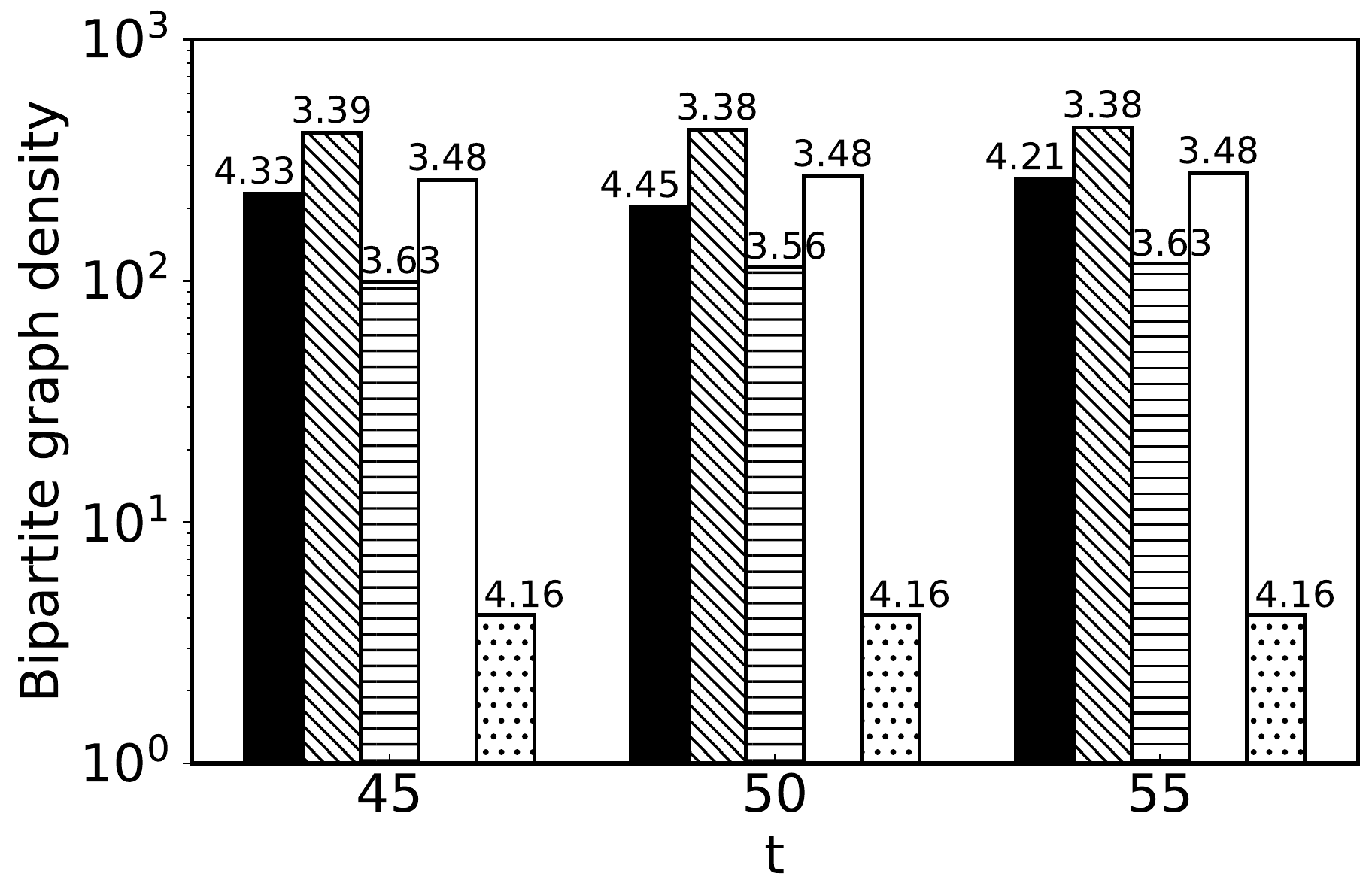}\vspace{-5.5mm}
\label{fig:effectiveness1}
\end{minipage}}
\subfigure[Percentage of dislike users]{
\begin{minipage}[b]{0.22\textwidth}
\includegraphics[trim=0 0 0 0,clip,width=1\textwidth]{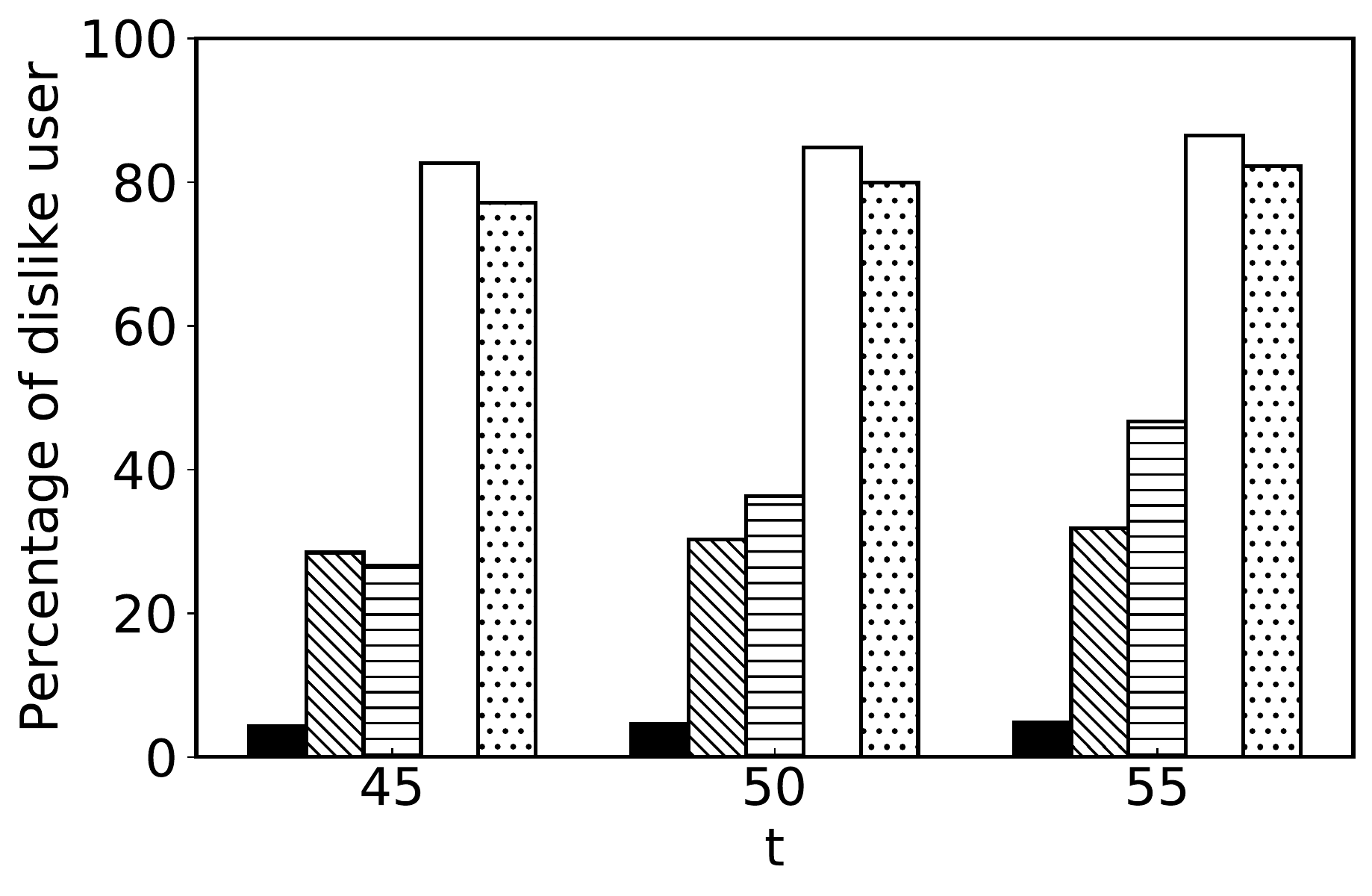}\vspace{-5.5mm}
\label{fig:effectiveness2}
\end{minipage}}
\vspace{-3mm}
\caption{{\color{black}Evaluating the community quality, varying $\alpha,\beta=t$}}
\label{fig:effectiveness}
\end{centering}
\end{figure}

\noindent
{\bf Evaluating the community quality.} Suppose a user wants to find some friends who are also fans of comedy movies. We extract the subgraph formed by the ratings on comedy movies and perform community search algorithms. Figure \ref{fig:effectiveness1} shows the bipartite graph density  which is computed as $d(G)={|E(G)|/}{\sqrt{|U(G)||L(G)|}}$ \cite{kannan1999analyzing}. We can see that the communities produced by \abcore, bitruss, biclique and SC all have high densities comparing with $C_{4\star}$ since the structure cohesiveness is considered in these models. Thus, the users in $C_{4\star}$ are loosely connected with each other and have fewer interactions. In addition, the average ratings (i.e., the numbers on the top of each bar) indicate that SC can always return a group of users with higher average ratings than \abcore, bitruss and biclique. We also show the number of dislike users in Figure \ref{fig:effectiveness2}. A user is a dislike user if he/she gives fewer than $0.6\alpha$ good ratings (i.e., rating $\geq 4$), who is not likely to be a fan of comedies. We can see that SC contains fewer number of dislike users comparing with all the other models because both weight and structure cohesiveness are considered. Thus, the users in SC are considered as good candidates to be recommended to the query user. Note that the percentage of dislike users in bitruss and $C_{4\star}$ is very high. This is because bitruss ensures the structure cohesiveness using the butterfly (i.e., $2\times2$-biclique) and a user can exist in a $k$-bitruss with a large $k$ value if he/she only watched a few number of hot movies. In addition, $C_{4\star}$ does not ensure the structure cohesiveness and there exist many users who only watched few high rating movies. 

\begin{figure}[htb]
\begin{centering}
\includegraphics[trim=0 0 0 0,clip,width=0.42\textwidth]{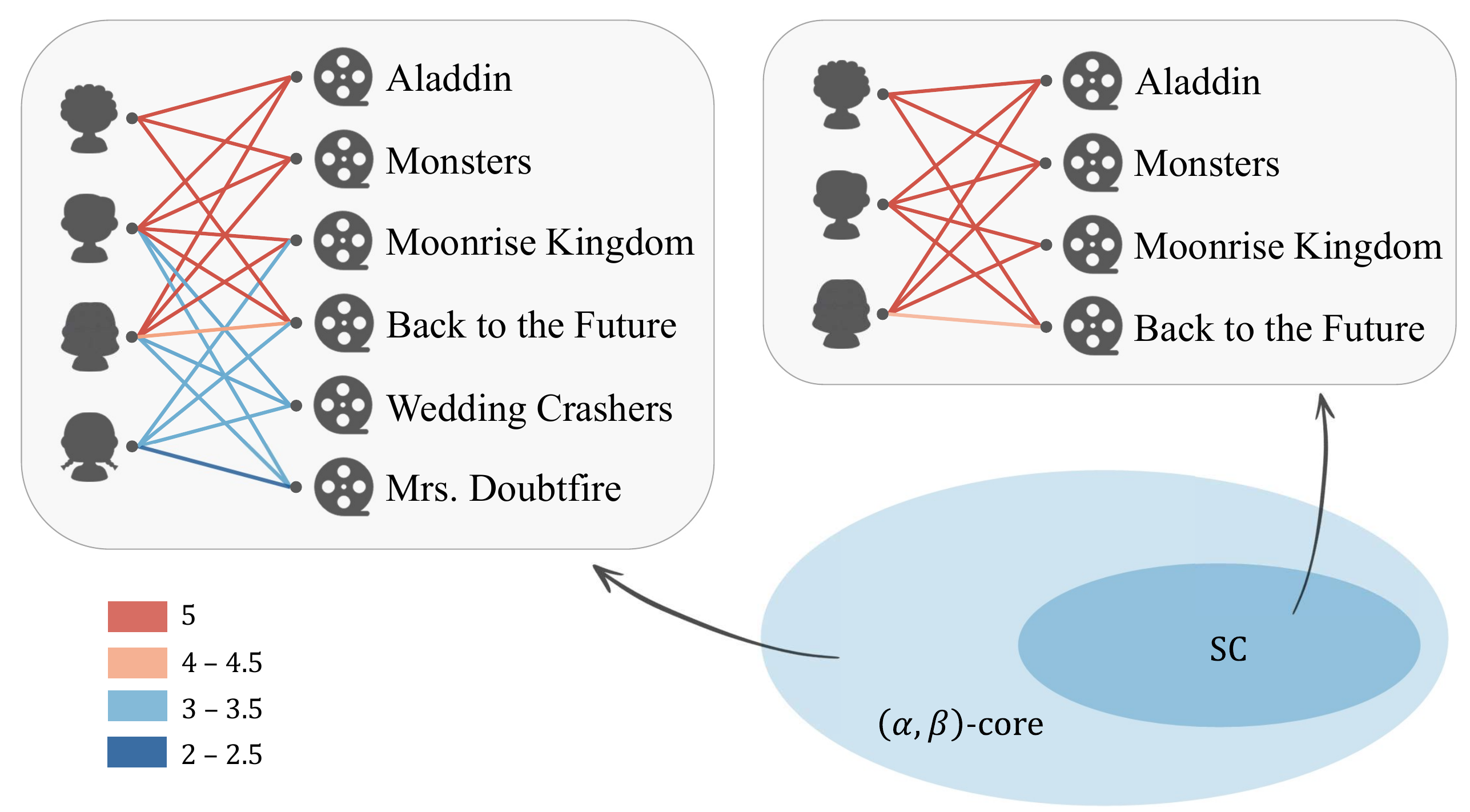}
\caption{{\color{black}Representative components of real-life communities}}
\label{fig:case}
\vspace*{-2mm}
\end{centering}
\end{figure}

\begin{table}[ht]
\small
\caption{{\color{black}Statistics of query results, $q=6{,}778$}}
\vspace*{-3mm}
\begin{center}
\scalebox{0.79}{
\begin{tabular}{c||c|c|c|c|c|c}
\noalign{\hrule height 1pt}
\textbf{Models} & $|U|$ & $|M|$ & $R_{avg}$  & $R_{min}$ & $M_{avg}$ & $Sim$ ($\%$)\\\hline
 SC & 2{,}127 & 670 & 4.81 & 4.50 & 63.47 & 100 \\ \hline
 $(\alpha,\beta)$-core & 34{,}466 & 2491 & 3.39 & 0.5 & 110.03 &  7.57 \\ \hline
 bitruss & 158{,}183 & 2{,}985 & 3.48 & 0.5 & 35.87 & 1.74\\ \hline
 biclique & 65 & 45 & 3.45 & 0.5 & 45 & 2.39\\ \hline
 $C_{4\star}$ & 114{,}915 & 387 & 4.16 & 0.5& 2.39 & 1.82\\
\noalign{\hrule height 1pt}
\end{tabular}}
\end{center}
\label{table:case}
\end{table}

\noindent
{\bf Case study.} We conduct queries using parameters $q=6778, \alpha=45, \beta=45$ on comedy movies. The statistics of query results are shown in Table \ref{table:case}. $|U|$ and $|M|$ denote the total number of users and movies in the community, respectively. $R_{avg}$ and $R_{min}$ denote the average and minimal rating in the community, respectively. $M_{avg}$ is the average number of movies a user watched in the community and $Sim$ is the jaccard similarity between each community and SC. For the biclique model, here we use a maximal biclique containing $q$ with at least 45 vertices in each layer. We can see that SC contains reasonable number of users and vertices with higher average rating and minimal rating in the community than the others. We also show the representative components of the communities using \abcore and SC in Figure \ref{fig:case}. We can see that \abcore contains users who do not like such movies and movies that are not liked by such users. This is because \abcore only considers structure cohesiveness and ignores the edge weights. We can observe that $M_{avg}$ of $C_{4\star}$ is only 2.39 since the structure cohesiveness is not considered in $C_{4\star}$. Thus, $C_{4\star}$ contains many users who only watched a few number of high rating movies and these users are loosely connected with the query user. Among these models, only SC considers both weight and structure cohesiveness, which is not similar to other communities compared here. In SC, each user has given at least 45 times 4.5-star ratings on these comedy movies and the movies are reviewed as 4.5-star at least 45 times by the users. Thus, the quality of the users and movies found by SC can be guaranteed and highly recommended to the query user. 
}

\begin{figure}[!h]
\vspace*{-2mm}
\begin{centering}
\includegraphics[trim=0 0 0 0,clip,width=0.48\textwidth]{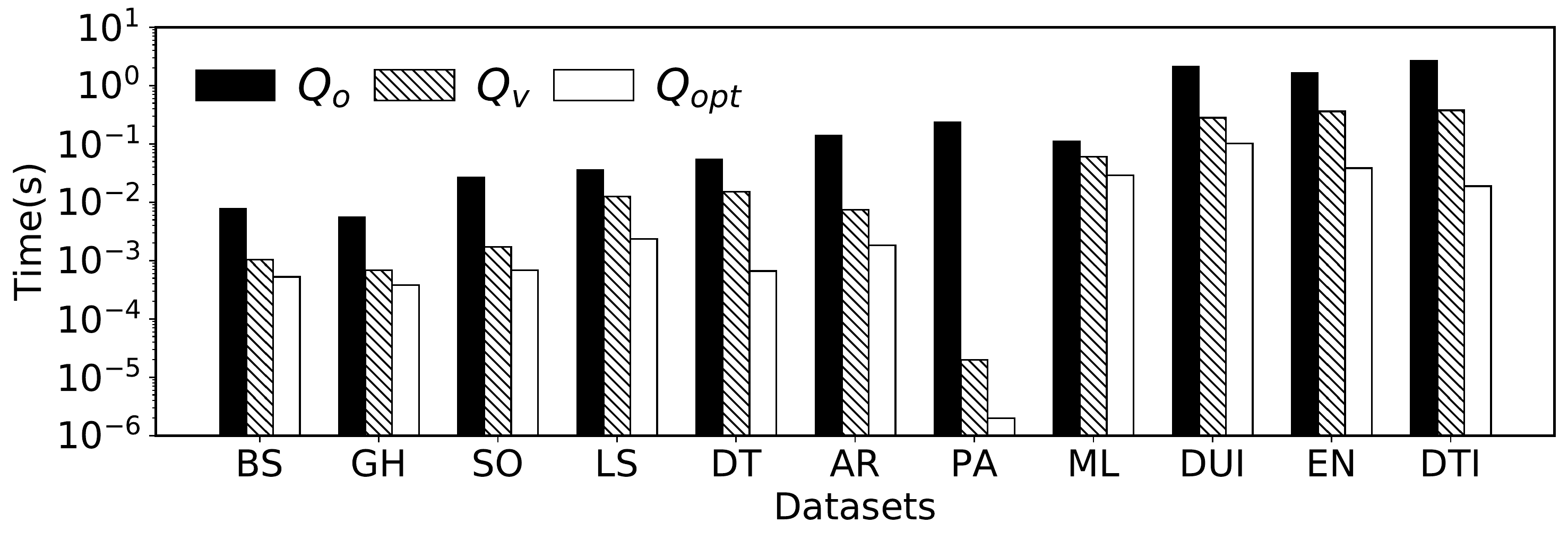}
\vspace*{-1mm}
\caption{Retrieving the \abcorecoms}
\label{fig:abcorec}
\end{centering}
\end{figure}


\begin{figure}[htb]
\begin{centering}
\subfigure[\texttt{EN}, $\alpha, \beta = c \cdot \delta$]{
\begin{minipage}[b]{0.2\textwidth}
\includegraphics[trim=0 0 0 0,clip,width=1\textwidth]{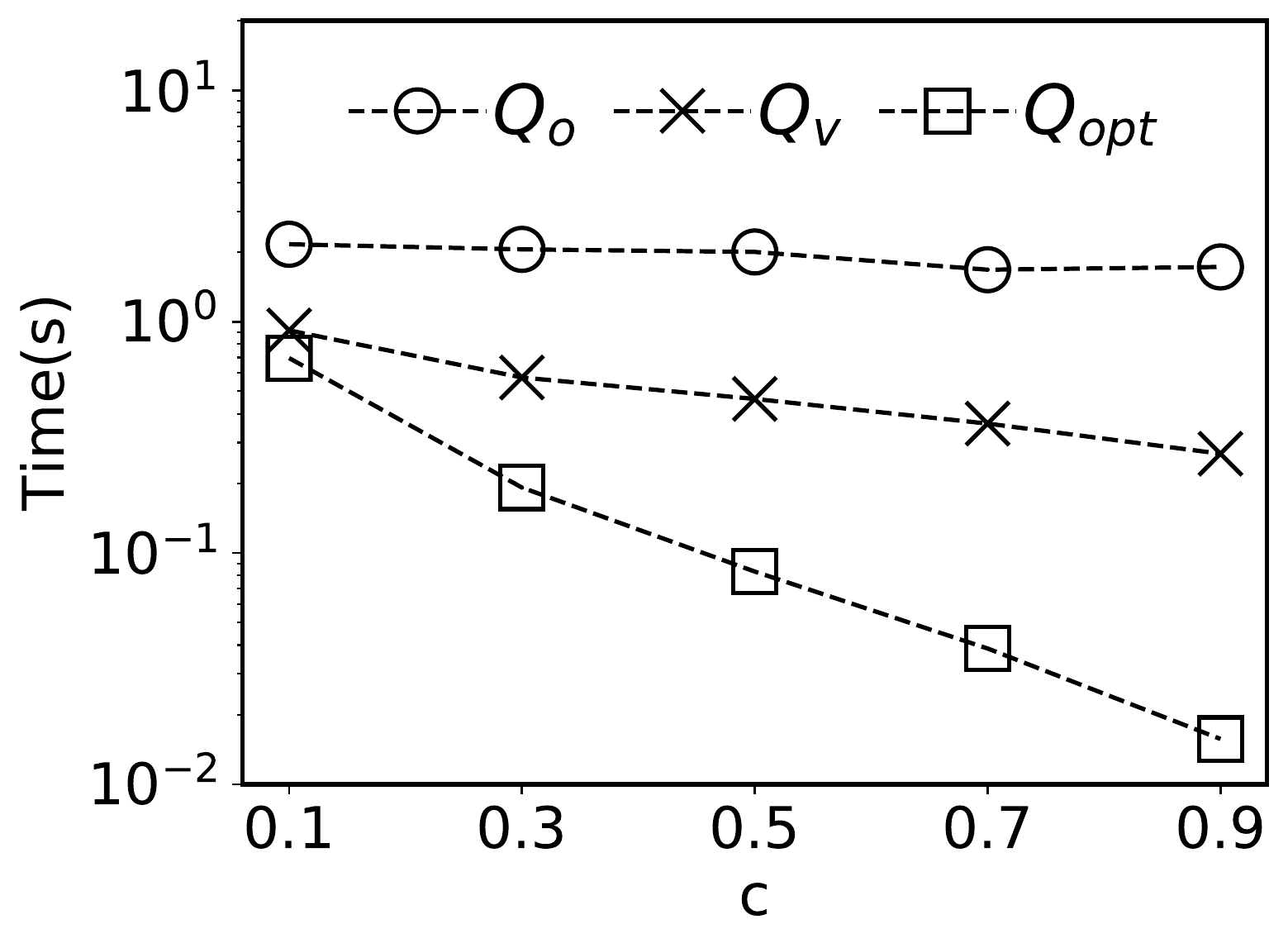}\vspace{-5.5mm}
\label{fig:p1}
\end{minipage}}
\subfigure[\texttt{SO}, $\alpha, \beta = c \cdot \delta$]{
\begin{minipage}[b]{0.2\textwidth}
\includegraphics[trim=0 0 0 0,clip,width=1\textwidth]{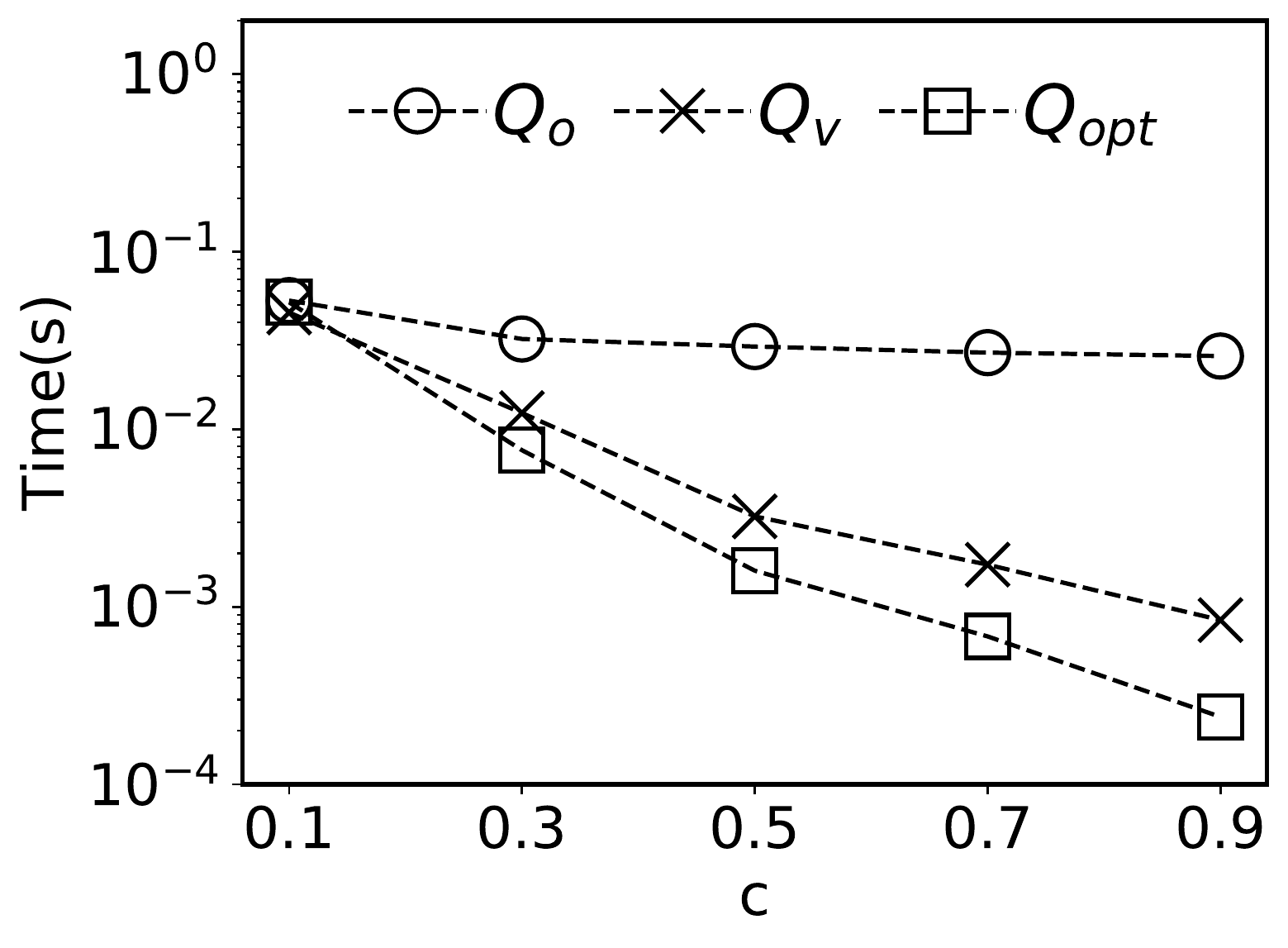}\vspace{-5.5mm}
\label{fig:p2}
\end{minipage}}
\subfigure[\texttt{EN}, $\alpha = 0.5 \cdot \delta$, $\beta= c \cdot \delta$]{
\begin{minipage}[b]{0.2\textwidth}
\includegraphics[trim=0 0 0 0,clip,width=1\textwidth]{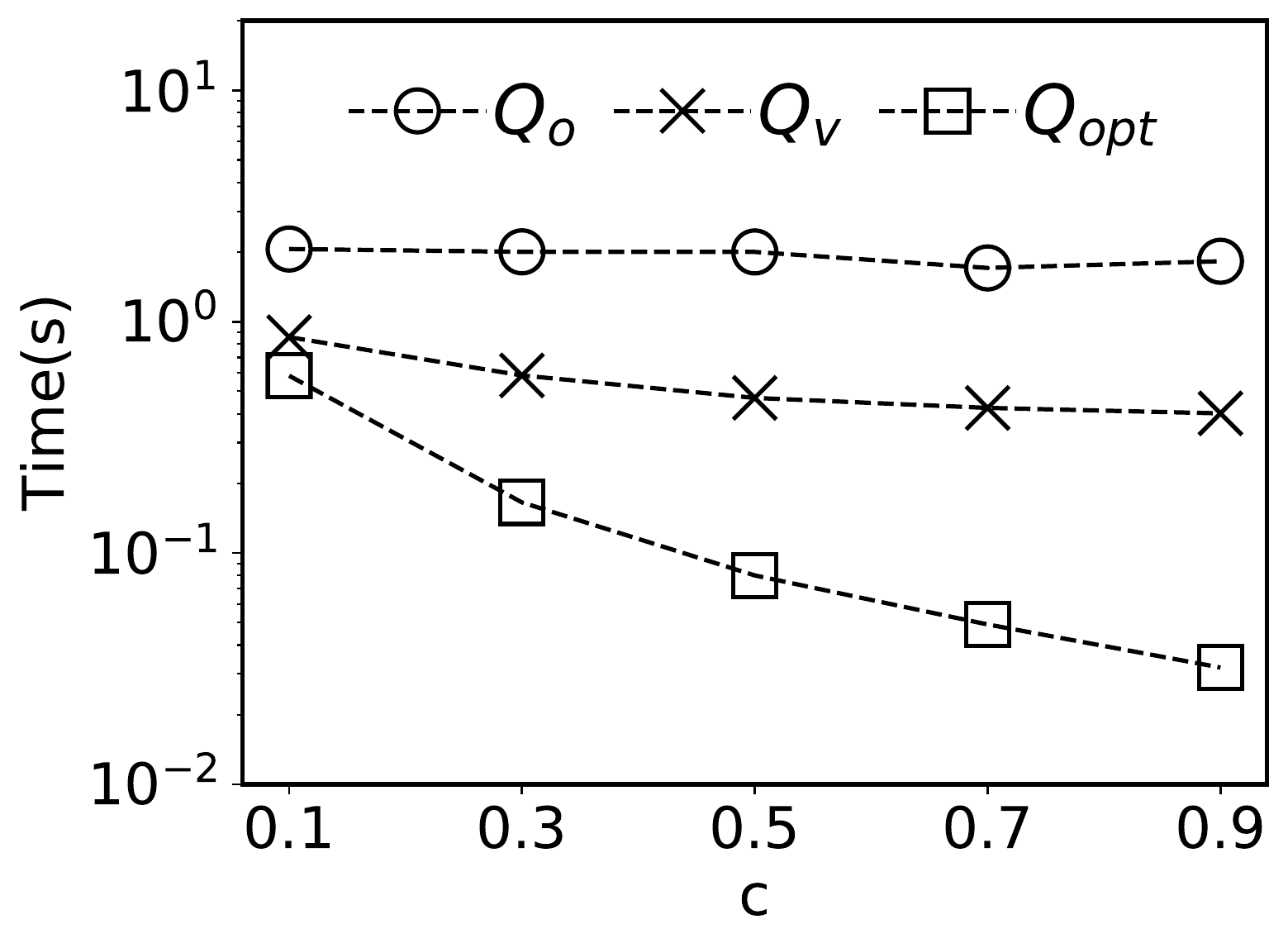}\vspace{-5.5mm}
\label{fig:p3}
\end{minipage}}         
\subfigure[\texttt{SO}, $\alpha= c \cdot \delta$, $\beta = 0.5 \cdot \delta$]{
\begin{minipage}[b]{0.2\textwidth}
\includegraphics[trim=0 0 0 0,clip,width=1\textwidth]{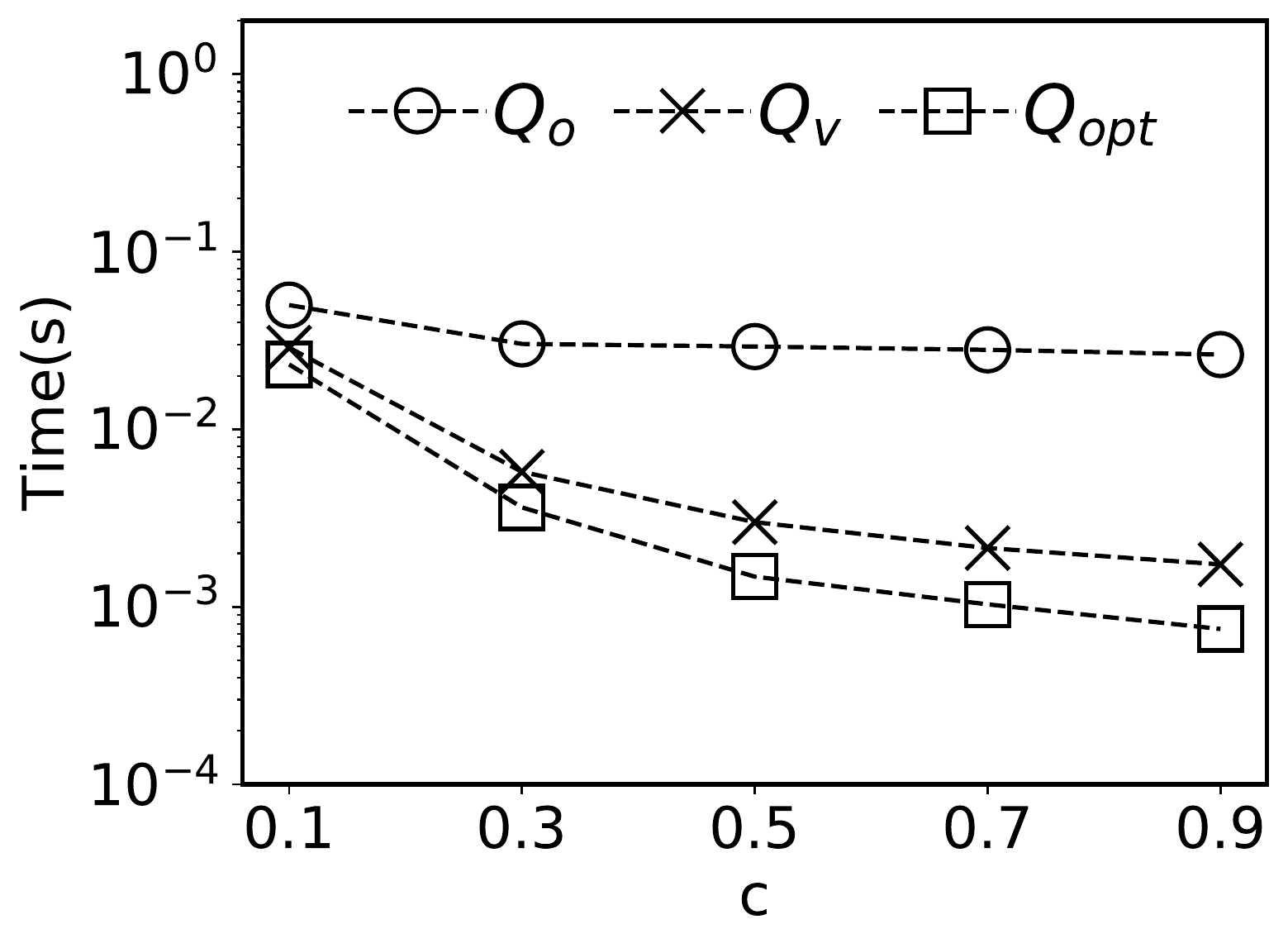}\vspace{-5.5mm}
\label{fig:p4}
\end{minipage}}
\vspace*{-3mm}\caption{{\color{black}Retrieving the \abcorecoms, varying $\alpha$ and $\beta$}}
\label{fig:vary_component_query}
\end{centering}
\end{figure}

\subsection{Evaluation of retrieving \abcorecom}
In this part, we evaluate the proposed indexing techniques to retrieve the \abcorecom.

\noindent
{\bf Query time.} {\em 1) Performance on all the datasets.} We first evaluate the performance on all the datasets by setting $\alpha$ and $\beta$ to $0.7\delta$. In Figure \ref{fig:abcorec}, we can observe that $Q_{opt}$ significantly outperforms $Q_{o}$ and $Q_{v}$ on all the datasets. This is because $Q_{opt}$ is based on $\indexad$ which can achieve optimal retrieval of \abcorecoms. Especially, on large datasets such as \texttt{DUI}, \texttt{EN} and \texttt{DTI}, the $Q_{opt}$ algorithm is one to two orders of magnitude faster than $Q_{o}$ and is up to 20$\times$ faster than $Q_{v}$.


\noindent
{\em 2) Varying $\alpha$ and $\beta$.} We also vary $\alpha$ and $\beta$ to assess the performance of these algorithms. {\color{black}In Figure \ref{fig:vary_component_query}(a) and (b), $\alpha$ and $\beta$ are varied simultaneously. We can observe that when $\alpha$ and $\beta$ are small, the performance of these algorithms is similar. This is because only a few number of edges are removed from the original graph when the query parameters are small. When $\alpha$ and $\beta$ are large, the resulting \abcorecoms are much smaller than the original graph. Thus, $Q_{opt}$ is much faster than $Q_{o}$ and $Q_{v}$. In Figure \ref{fig:vary_component_query}(c) and (d), we fix $\alpha$ (or $\beta$) and vary the other one and the trends are similar.}

\begin{figure}[!h]
\begin{centering}
\includegraphics[trim=0 0 0 0,clip,width=0.49\textwidth]{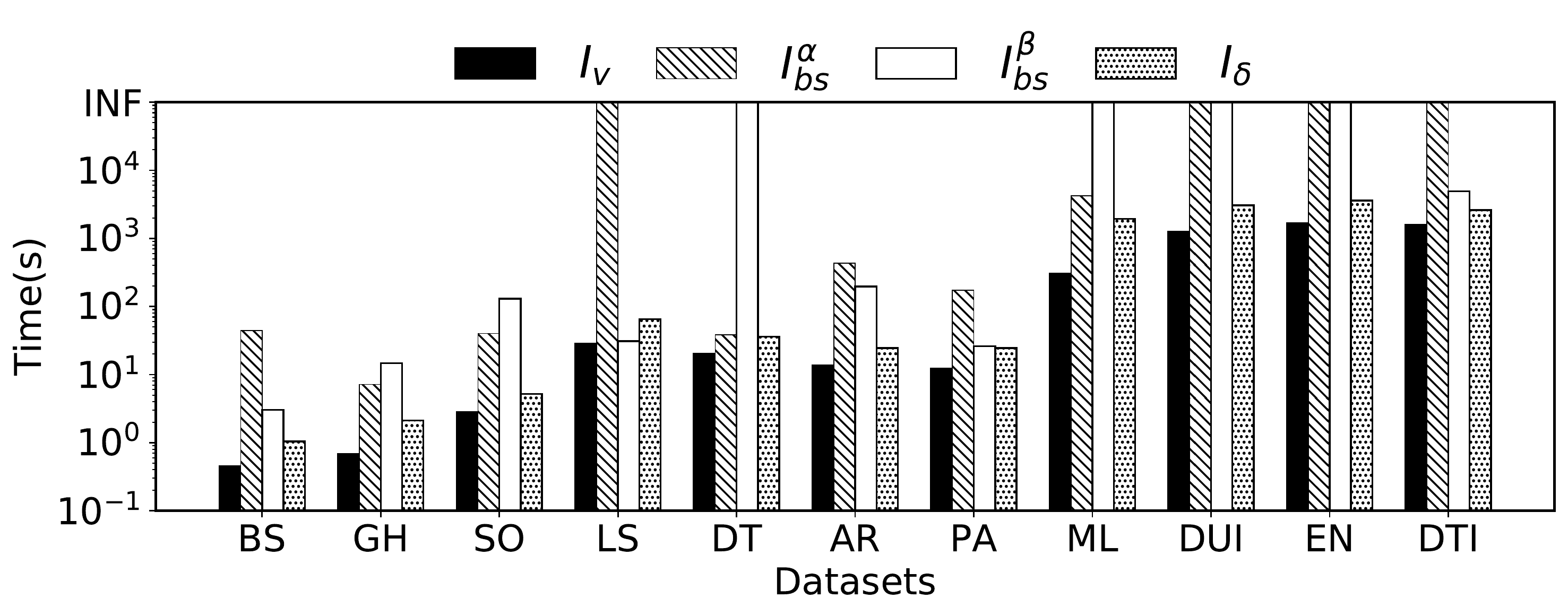}
\vspace{-6mm}
\caption{Index construction time}
\label{fig:index_time}
\end{centering}
\end{figure}

\begin{figure}[!h]
\begin{centering}
\includegraphics[trim=0 0 0 0,clip,width=0.49\textwidth]{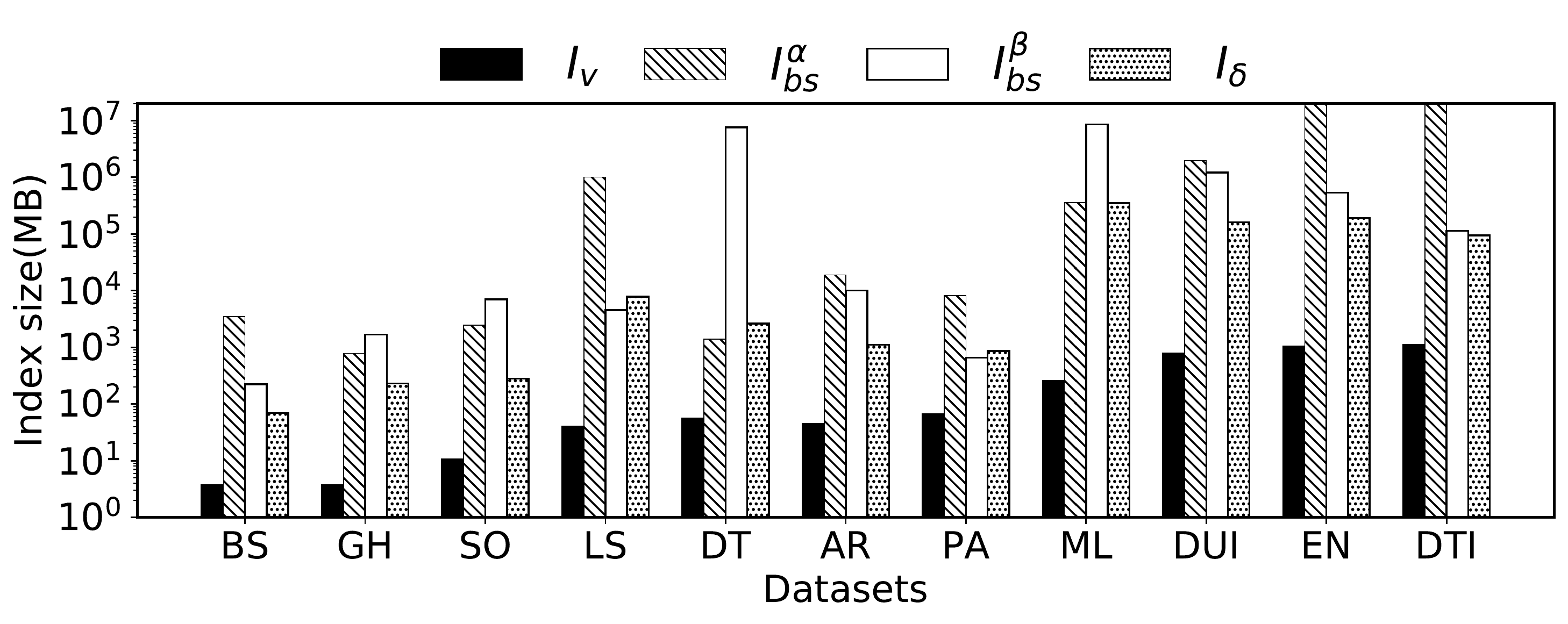}
\vspace{-6mm}
\caption{Index size}
\label{fig:index_size}
\end{centering}
\end{figure}

\noindent
{\bf Evaluating index construction time and index size.} In this part, we evaluate the index size and index construction time. 

\noindent
{\em 1) Index construction time.} In Figure \ref{fig:index_time}, we can see that $\indexad$ can be efficiently constructed on all the datasets since it only needs the same low constructing time complexity as $\indexbc$ ($O(\delta m)$). In addition, constructing $\indexad$ is slightly slower than constructing $\indexbc$ which is reasonable since $\indexbc$ only contains vertex information of \abcores while $\indexad$ contains edge information which can support optimal retrieval of \abcorecoms. The time for constructing $\indexbsa$ and $\indexbsb$ highly depends on $\alpha_{max}$ and $\beta_{max}$. Thus, it is very slow (or even unaccomplished) on the datasets where these two values are large such as \texttt{DUI} and \texttt{EN}.

\noindent
{\em 2) Index size.} In Figure \ref{fig:index_size}, we evaluate the size of these indexes. If an index cannot be built within the time limit, we report the expected size of it. We can see that $\size(\indexad)$ is smaller than $\size(\indexbsa)$ and $\size(\indexbsb)$ on almost all the datasets. $\indexbc$ is the index with the minimal size since it only contains vertex information.



\subsection{Evaluation of retrieving \wcom}
Here we evaluate the performance of the algorithms (\baseline, \peel, and \expand) for querying \wcoms. In these algorithms, we use $Q_{opt}$ to support the optimal retrieval of \abcorecom. In each test, we randomly select 100 queries and take the average.

\begin{figure}[!h]
\begin{centering}
\includegraphics[trim=0 0 0 0,clip,width=0.49\textwidth]{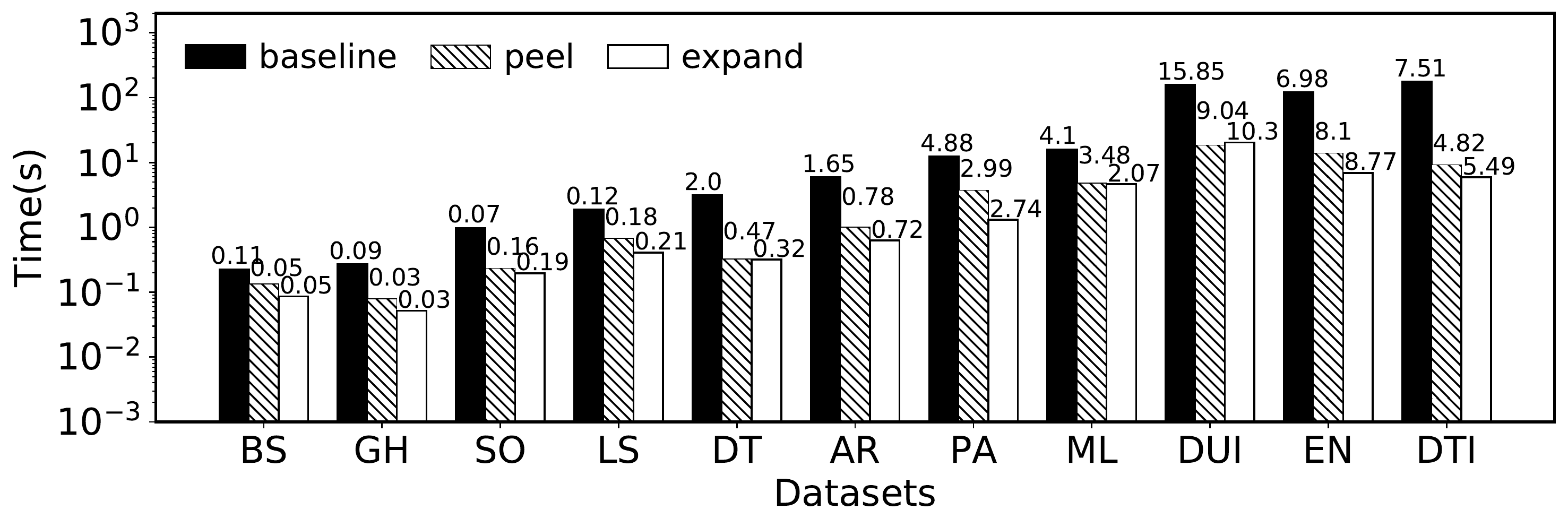}
\vspace*{-5mm}
\caption{{\color{black}Query performance on different datasets}}
\label{fig:wcom}
\end{centering}
\end{figure}

{\color{black}
\noindent
{\bf Evaluating the performance on all the datasets.}
In Figure \ref{fig:wcom}, we evaluate the performance of \baseline, \peel, and \expand on all the datasets. We also report the standard deviation on the top of each bar. We can see that \expand and \peel are significantly faster than \baseline, especially on large datasets. This is because, with the help of the two-step framework, the search space of \peel and \expand is limited in $\abcorecgq$, while \baseline needs to consider all edges in the connected component containing $q$ of the whole graph. We can also see in Table \ref{table:datasets} that $|R_{\delta,\delta}|$ is much smaller than $|E|$. Since $C_{\delta,\delta}(q) \subseteq R_{\delta,\delta}$, when we choose relatively larger parameters, the search space of \peel and \expand is much smaller than \baseline. In addition, we can see that on most datasets, \expand is on average more efficient than \peel. However, the standard deviations of \expand and \peel are large. This is because \peel and \expand both need more time to handle the cases when $\alpha$ and $\beta$ are small and \expand is usually much faster than \peel.}

\begin{figure}[htb]
\begin{centering}
\subfigure[\texttt{DT}, $\alpha, \beta = c \cdot \delta$]{
\begin{minipage}[b]{0.203\textwidth}
\includegraphics[trim=0 0 0 0,clip,width=1\textwidth]{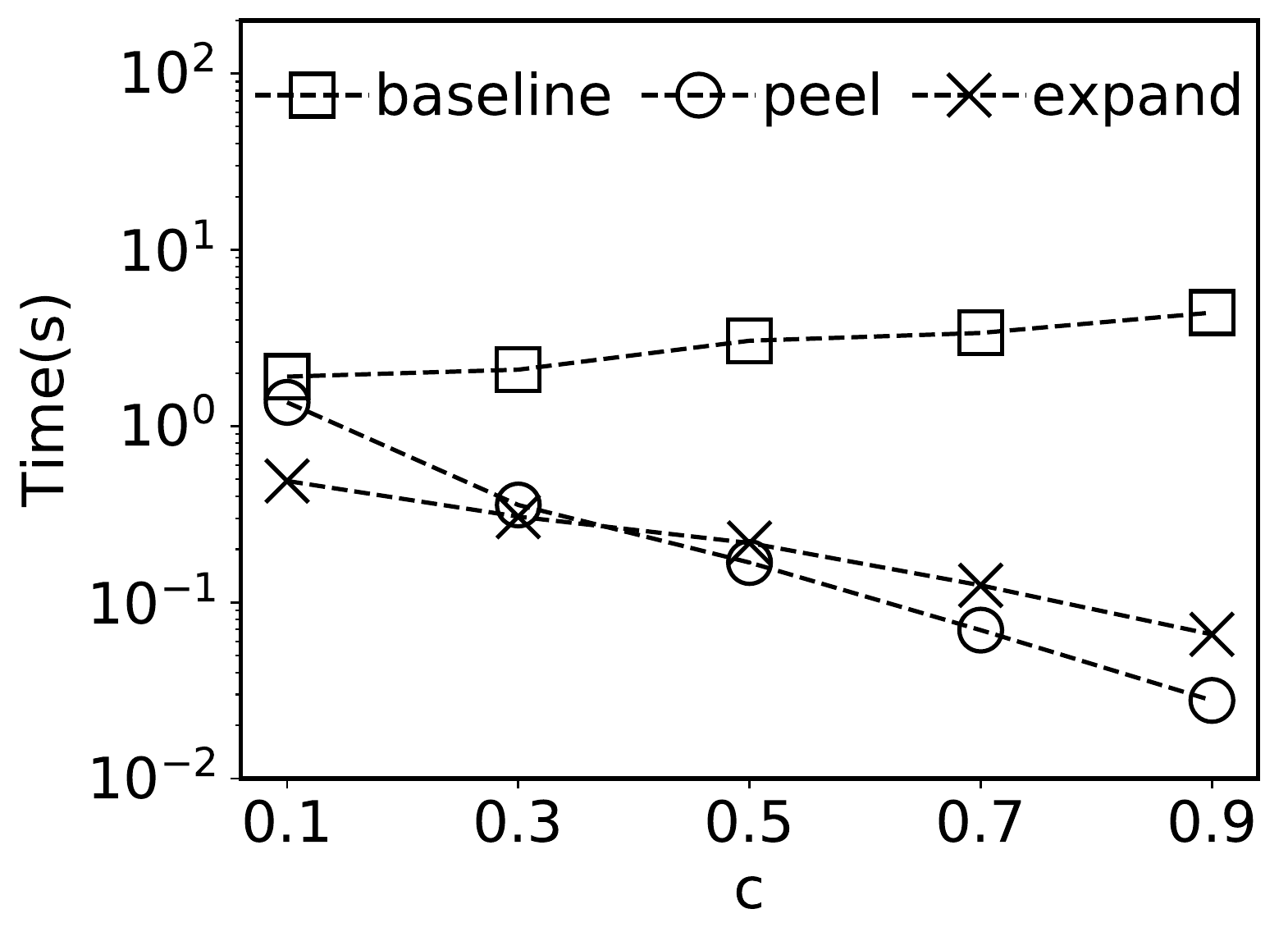}\vspace{-5.5mm}
\label{fig:p11}
\end{minipage}}
\subfigure[\texttt{ML}, $\alpha, \beta = c \cdot \delta$]{
\begin{minipage}[b]{0.2\textwidth}
\includegraphics[trim=0 0 0 0,clip,width=1\textwidth]{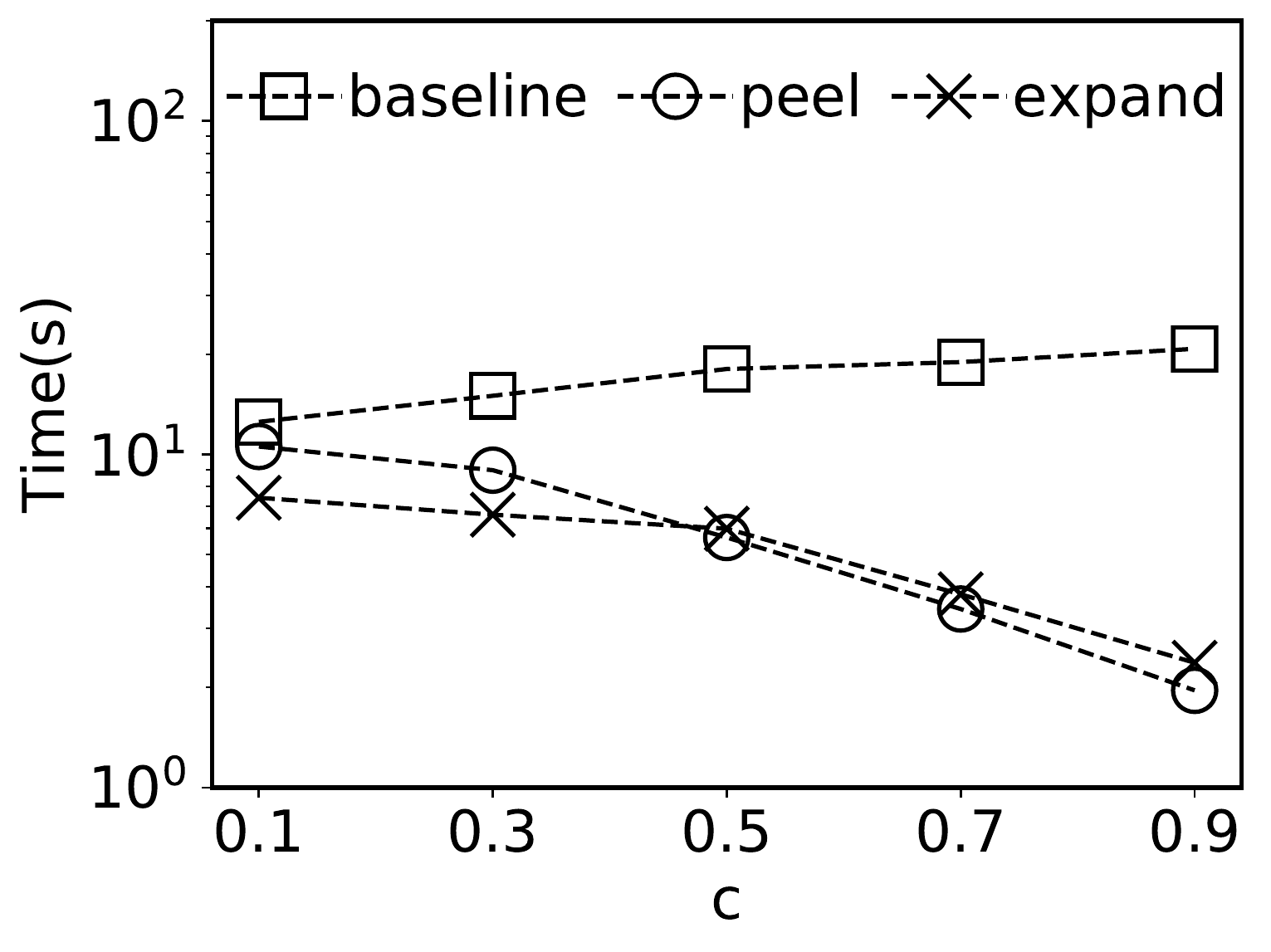}\vspace{-5.5mm}
\label{fig:p13}
\end{minipage}}
\subfigure[\texttt{DT}, {$\alpha = c \cdot \delta, \beta = 0.5 \cdot \delta$}]{
\begin{minipage}[b]{0.203\textwidth}
\includegraphics[trim=0 0 0 0,clip,width=1\textwidth]{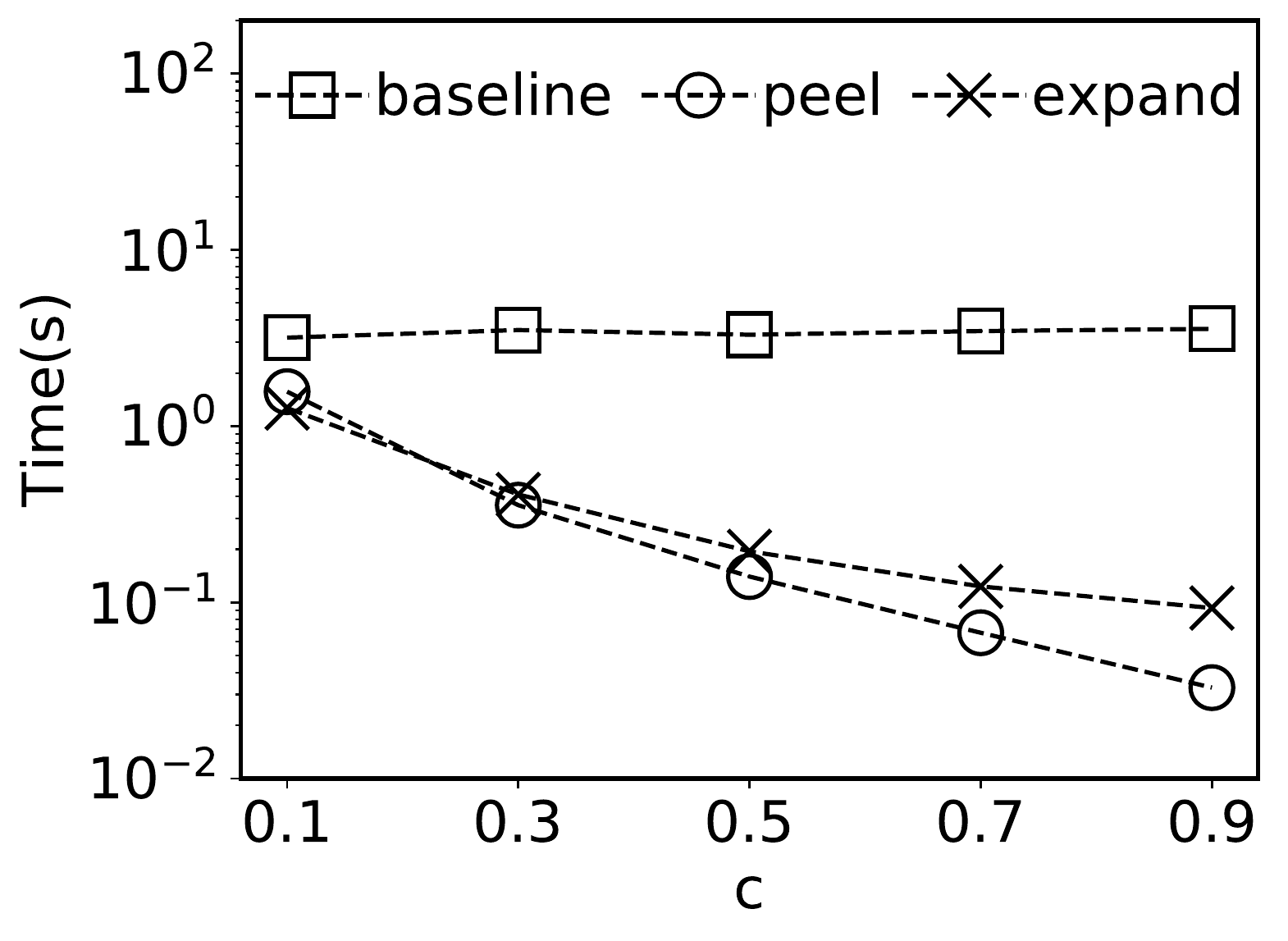}\vspace{-5.5mm}
\label{fig:p12}
\end{minipage}}
\subfigure[\texttt{ML}, {$\alpha = 0.5 \cdot \delta, \beta = c \cdot \delta$}]{
\begin{minipage}[b]{0.2\textwidth}
\includegraphics[trim=0 0 0 0,clip,width=1\textwidth]{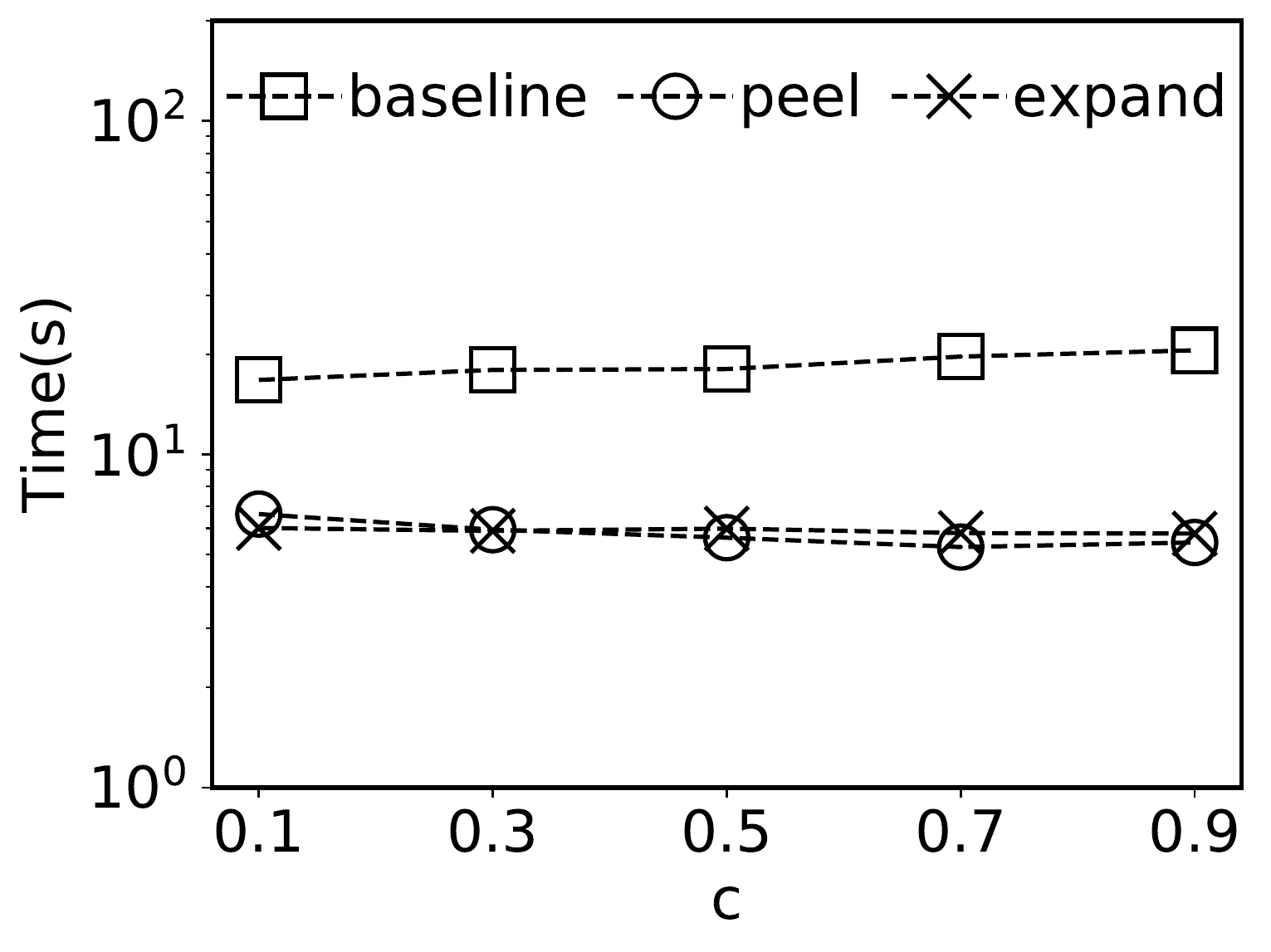}\vspace{-5.5mm}
\label{fig:p14}
\end{minipage}}
\vspace*{-2mm}\caption{{\color{black}Effect of $\alpha$ and $\beta$}}
\label{fig:ab}
\end{centering}
\end{figure}

{\color{black}
\noindent
{\bf Evaluating the effect of query parameters $\alpha$ and $\beta$.} In Figure \ref{fig:ab}, we vary $\alpha$ and $\beta$ on two datasets \texttt{DT} and \texttt{ML}. From Figure \ref{fig:ab}(a) and (b), we can see that, when $\alpha$ and $\beta$ are small, \expand is more efficient than \peel. In addition, the running time of \peel and \expand decreases as $\alpha$ (or $\beta$) increases. Note that the efficiency of these two algorithms largely depends on the size of the $(\alpha, \beta)$-community containing $q$ (i.e., $\size(\abcorecgq)$, which determines the search space) and the size of the final result (i.e., $\size(\wcomg)$, which relates to the actual computation cost). In most cases, when $\alpha$ and $\beta$ are large, the size of $\abcorecgq$ is small and $\wcomg$ is expected to be large since more edges are needed in $\wcomg$ to satisfy the cohesiveness constraints. Thus, the edges need to be peeled are usually few and \peel is more efficient than \expand. When $\alpha$ and $\beta$ are small, the search space (i.e., $\abcorecgq$) can be large and $\wcomg$ is expected to be small. Thus, \expand is usually more efficient than \peel in these cases. In most cases, we can determine to use \peel or \expand according to the choice of $\alpha$ and $\beta$.}

\begin{table}[ht]
\small
\caption{\label{table:distribution}
{\color{black}Running time under different weight distribution}}
\vspace{-3.5mm}
\begin{center}
\scalebox{1.0}{
\begin{tabular}{c||c|c|c|c}
\noalign{\hrule height 1pt}
\textbf{Algorithms} & AE & RW & UF & SK \\\hline
 \baseline & 0.03s & 3.12s & 4.42s & 4.31s \\ \hline
 \peel & 0.03s & 0.34s & 0.48s & 0.45s \\ \hline
 \expand & 0.03s & 0.31s & 0.41s  & 0.36s \\ \hline
\noalign{\hrule height 1pt}
\end{tabular}}
\end{center}
\end{table}

{\color{black}
\noindent
{\bf Evaluate the effect of weight distribution.}
In Table \ref{table:distribution}, we evaluate the effect of weight distribution on \texttt{DT} dataset. We test four weight distributions: (1) AE: the weights are all equal; (2) RW: the weights are generated using the random walk with restart model \cite{tong2006fast}; (3) UF: the weights follow uniform distribution; (4) SK: the weights follow skewed normal distribution with skewness = 1.02. When all the edge weights are equal (AE) which can be considered as a special case, all three algorithms can just return $\abcorecgq$ after efficiently scanning $\abcorecgq$. Note that the performances of these three algorithms are not very sensitive to the other three distributions. This is because both weight and structure cohesiveness are considered in our problem and the impact of RW/SK/UF weight distributions are limited.}



\section{Related Work}
\label{sct:related}
To the best of our knowledge, this paper is the first to study community search over bipartite graphs. Below we review two closely related areas, community search on unipartite graphs and cohesive subgraph models on bipartite graphs. 

\noindent
{\em Community search on unipartite graphs.} On unipartite graphs, community search is conducted based on different cohesiveness models such as $k$-core \cite{CXWW14SIGMOD, SG10KDD, barbieri2015efficient, fang2017effective, wang2018efficient, fang2016effective,ghafouri2020efficient, zhang2020exploring,liu2019corecube,fang2018effective,fang2018spatial,zhang2017finding,zhang2017olak}, $k$-truss \cite{akbas2017truss, huang2014querying, huang2015approximate, zheng2017finding,liu2021efficient,zhang2018finding}, clique \cite{yuan2017index,fang2019efficient,yuan2017index}. Interested readers can refer to \cite{CSSurvey2020} for a recent comprehensive survey.

{\color{black} Based on $k$-core, \cite{CXWW14SIGMOD} and \cite{SG10KDD} study online algorithms for $k$-core community search on unipartite graphs. In \cite{barbieri2015efficient}, Barbieri et al. propose a tree-like index structure for the $k$-core community search. Using $k$-core, Fang et al. \cite{fang2016effective} further integrate the attributes of vertices to identify community and the spatial locations of vertices are considered in \cite{fang2017effective, wang2018efficient}. For the truss-based community search, \cite{akbas2017truss, huang2014querying} study the triangle-connected model and \cite{huang2015approximate} studies the closest model. In \cite{yuan2017index}, the authors study the problem of densest clique percolation community search. However, the edge weights are not considered in any of the above works and their techniques cannot be easily extended to solve our problem. On edge-weighted unipartite graphs, the $k$-core model is applied to find cohesive subgraphs in \cite{garas2012k,eidsaa2013s}. They use a function to associate the edge weights with vertex degrees and the edge weights are not considered as a second factor apart from the graph structure. Thus, these works do not aim to find a cohesive subgraph with both structure cohesiveness and high weight (significance). Under their settings, a subgraph with loose structure can be found in the result. For example, a vertex can be included in the result if it is only incident with one large-weight edge. In \cite{zheng2017finding}, the $k$-truss model is adopted on edge-weighted graphs to find communities. However, the $k$-truss model is based on the triangle structure which does not exist on bipartite graphs. One may also consider using the graph projection technique \cite{newman2001scientific} to generate a unipartite projection from the original (weighted) bipartite graph. The drawback of this approach is twofold. Firstly, it can cause information loss and edge explosion \cite{sariyuce2018peeling}. Secondly, it is not easy to project a weighted bipartite graph and handle the projected graph using existing methods. This is because we need to consider two kinds of weights (i.e., the original edge weight and the structure weight generated from another layer) on the projected graph.} 

\noindent
{\em Finding cohesive subgraphs on bipartite graphs.} On bipartite graphs, several existing works \cite{ding2017efficient,liu2019,he2020exploring,liu2020efficient} extend the $k$-core model on unipartite graph to the $(\alpha, \beta)$-core model.
Based on the butterfly structure \cite{wang2019vertex}, \cite{zou2016bitruss, sariyuce2018peeling, wang2020efficient} study the bitruss model in bipartite graphs which is the maximal subgraph where each edge is contained in at least $k$ butterflies.
\cite{zhang2014finding} studies the biclique enumeration problem. {\color{black} However, the above works only consider the structure cohesiveness and ignore the edge weights which are important as validated in the experiments. In the literature, fair clustering problems \cite{chierichetti2017fair,ahmadi2020fair,ahmadian2020fair} are studied to find communities (i.e., clusters) under fairness constraints on bipartite graphs. The problem is inherently different and the techniques are not applicable to the problem studied in this paper. An interesting work in \cite{kobren2019paper} studies the paper matching problem in peer-review process which also finds dense subgraphs on bipartite graphs. However, their flow-based techniques are often used to solve a matching problem while our problem is not modeled as a matching problem.}

\vspace{-0.1cm}
\section{Conclusion}

\vspace{-0.1cm}
\label{sct:conclusion}
In this paper, we study the \wcom search problem. To solve this problem efficiently, we follow a two-step framework which first retrieves the \abcorecom, and then identifies the \wcom from the \abcorecom. We develop a novel index $\indexad$ to retrieve the \abcorecom in optimal time. In addition, we propose efficient peeling and expansion algorithms to obtain the \wcom. 
We conduct extensive experiments on real-world graphs, and the results demonstrate the effectiveness of the \wcom model and the proposed techniques. 

\vspace{-0.2cm}

\section{Acknowledgment}
\vspace{-0.1cm}

\label{sct:acknowledgment}
{\color{black}Xuemin Lin is supported by NSFC61232006, 2018YFB1003504, ARC DP200101338,  ARC DP180103096 and ARC DP170101628. Lu Qin is supported by ARC FT200100787. Wenjie Zhang is supported by ARC DP180103096 and ARC DP200101116. Ying Zhang is supported by FT170100128 and ARC DP180103096.} We would like to thank Yizhang He for his proofreading.


{
\bibliographystyle{IEEEtran}
\bibliography{paper}

\begin{thebibliography}{10}
\providecommand{\url}[1]{#1}
\csname url@samestyle\endcsname
\providecommand{\newblock}{\relax}
\providecommand{\bibinfo}[2]{#2}
\providecommand{\BIBentrySTDinterwordspacing}{\spaceskip=0pt\relax}
\providecommand{\BIBentryALTinterwordstretchfactor}{4}
\providecommand{\BIBentryALTinterwordspacing}{\spaceskip=\fontdimen2\font plus
\BIBentryALTinterwordstretchfactor\fontdimen3\font minus
  \fontdimen4\font\relax}
\providecommand{\BIBforeignlanguage}[2]{{%
\expandafter\ifx\csname l@#1\endcsname\relax
\typeout{** WARNING: IEEEtran.bst: No hyphenation pattern has been}%
\typeout{** loaded for the language `#1'. Using the pattern for}%
\typeout{** the default language instead.}%
\else
\language=\csname l@#1\endcsname
\fi
#2}}
\providecommand{\BIBdecl}{\relax}
\BIBdecl

\bibitem{wang2006unifying}
J.~Wang, A.~P. De~Vries, and M.~J. Reinders, ``Unifying user-based and
  item-based collaborative filtering approaches by similarity fusion,'' in
  \emph{SIGIR}.\hskip 1em plus 0.5em minus 0.4em\relax ACM, 2006, pp. 501--508.

\bibitem{beutel2013copycatch}
A.~Beutel, W.~Xu, V.~Guruswami, C.~Palow, and C.~Faloutsos, ``Copycatch:
  stopping group attacks by spotting lockstep behavior in social networks,'' in
  \emph{WWW}.\hskip 1em plus 0.5em minus 0.4em\relax ACM, 2013, pp. 119--130.

\bibitem{konect:DBLP}
M.~Ley, ``The {DBLP} computer science bibliography: Evolution, research issues,
  perspectives,'' in \emph{Proc. Int. Symposium on String Processing and
  Information Retrieval}, 2002, pp. 1--10.

\bibitem{CXWW14SIGMOD}
W.~Cui, Y.~Xiao, H.~Wang, and W.~Wang, ``Local serach of communities in large
  graphs,'' in \emph{SIGMOD}, 2014, pp. 991--1002.

\bibitem{SG10KDD}
M.~Sozio and A.~Gionis, ``The community-search problem and how to plan a
  succesful cocktail party,'' in \emph{SIGKDD}, 2010, pp. 939--948.

\bibitem{barbieri2015efficient}
N.~Barbieri, F.~Bonchi, E.~Galimberti, and F.~Gullo, ``Efficient and effective
  community search,'' \emph{Data mining and knowledge discovery}, vol.~29,
  no.~5, pp. 1406--1433, 2015.

\bibitem{fang2016effective}
Y.~Fang, R.~Cheng, S.~Luo, and J.~Hu, ``Effective community search for large
  attributed graphs,'' \emph{PVLDB}, vol.~9, no.~12, pp. 1233--1244, 2016.

\bibitem{huang2014querying}
X.~Huang, H.~Cheng, L.~Qin, W.~Tian, and J.~X. Yu, ``Querying k-truss community
  in large and dynamic graphs,'' in \emph{SIGMOD}, 2014, pp. 1311--1322.

\bibitem{huang2015approximate}
X.~Huang, L.~V.~S. Lakshmanan, J.~X. Yu, and H.~Cheng, ``Approximate closest
  community search in networks,'' \emph{PVLDB}, vol.~9, no.~4, pp. 276--287,
  2015.

\bibitem{huang2017attribute}
X.~Huang and L.~V. Lakshmanan, ``Attribute-driven community search,''
  \emph{PVLDB}, vol.~10, no.~9, pp. 949--960, 2017.

\bibitem{CSSurvey2020}
Y.~Fang, X.~Huang, L.~Qin, Y.~Zhang, R.~Cheng, and X.~Lin, ``A survey of
  community search over big graphs,'' \emph{VLDB J.}, vol.~29, no.~1, pp.
  353--392, 2020.

\bibitem{chierichetti2017fair}
F.~Chierichetti, R.~Kumar, S.~Lattanzi, and S.~Vassilvitskii, ``Fair clustering
  through fairlets,'' in \emph{Advances in Neural Information Processing
  Systems}, 2017, pp. 5029--5037.

\bibitem{ahmadi2020fair}
S.~Ahmadi, S.~Galhotra, B.~Saha, and R.~Schwartz, ``Fair correlation
  clustering,'' \emph{arXiv preprint arXiv:2002.03508}, 2020.

\bibitem{ahmadian2020fair}
S.~Ahmadian, A.~Epasto, M.~Knittel, R.~Kumar, M.~Mahdian, B.~Moseley, P.~Pham,
  S.~Vassilvtiskii, and Y.~Wang, ``Fair hierarchical clustering,'' \emph{arXiv
  preprint arXiv:2006.10221}, 2020.

\bibitem{liu2019}
B.~Liu, L.~Yuan, X.~Lin, L.~Qin, W.~Zhang, and J.~Zhou, ``Efficient ($\alpha$,
  $\beta$)-core computation: An index-based approach,'' in \emph{WWW}.\hskip
  1em plus 0.5em minus 0.4em\relax ACM, 2019, pp. 1130--1141.

\bibitem{ding2017efficient}
D.~Ding, H.~Li, Z.~Huang, and N.~Mamoulis, ``Efficient fault-tolerant group
  recommendation using alpha-beta-core,'' in \emph{CIKM}, 2017, pp. 2047--2050.

\bibitem{wang2020efficient}
K.~Wang, X.~Lin, L.~Qin, W.~Zhang, and Y.~Zhang, ``Efficient bitruss
  decomposition for large-scale bipartite graphs,'' in \emph{ICDE}.\hskip 1em
  plus 0.5em minus 0.4em\relax IEEE, 2020, pp. 661--672.

\bibitem{zou2016bitruss}
Z.~Zou, ``Bitruss decomposition of bipartite graphs,'' in \emph{DASFAA}.\hskip
  1em plus 0.5em minus 0.4em\relax Springer, 2016, pp. 218--233.

\bibitem{sariyuce2018peeling}
A.~E. Sar{\i}y{\"u}ce and A.~Pinar, ``Peeling bipartite networks for dense
  subgraph discovery,'' in \emph{WSDM}.\hskip 1em plus 0.5em minus 0.4em\relax
  ACM, 2018, pp. 504--512.

\bibitem{zhang2014finding}
Y.~Zhang, C.~A. Phillips, G.~L. Rogers, E.~J. Baker, E.~J. Chesler, and M.~A.
  Langston, ``On finding bicliques in bipartite graphs: a novel algorithm and
  its application to the integration of diverse biological data types,''
  \emph{BMC bioinformatics}, vol.~15, no.~1, p. 110, 2014.

\bibitem{khaouid2015k}
W.~Khaouid, M.~Barsky, V.~Srinivasan, and A.~Thomo, ``K-core decomposition of
  large networks on a single pc,'' \emph{PVLDB}, vol.~9, no.~1, pp. 13--23,
  2015.

\bibitem{cormen2009introduction}
T.~H. Cormen, C.~E. Leiserson, R.~L. Rivest, and C.~Stein, \emph{Introduction
  to algorithms}.\hskip 1em plus 0.5em minus 0.4em\relax MIT press, 2009.

\bibitem{tong2006fast}
H.~Tong, C.~Faloutsos, and J.-Y. Pan, ``Fast random walk with restart and its
  applications,'' in \emph{ICDM}.\hskip 1em plus 0.5em minus 0.4em\relax IEEE,
  2006, pp. 613--622.

\bibitem{chen2020structsim}
X.~Chen, L.~Lai, L.~Qin, and X.~Lin, ``Structsim: Querying structural node
  similarity at billion scale,'' in \emph{ICDE}.\hskip 1em plus 0.5em minus
  0.4em\relax IEEE, 2020, pp. 1950--1953.

\bibitem{jeh2002simrank}
G.~Jeh and J.~Widom, ``Simrank: a measure of structural-context similarity,''
  in \emph{SIGKDD}, 2002, pp. 538--543.

\bibitem{kannan1999analyzing}
R.~Kannan and V.~Vinay, \emph{Analyzing the structure of large graphs}.\hskip
  1em plus 0.5em minus 0.4em\relax Rheinische
  Friedrich-Wilhelms-Universit{\"a}t Bonn Bonn, 1999.

\bibitem{fang2017effective}
Y.~Fang, R.~Cheng, X.~Li, S.~Luo, and J.~Hu, ``Effective community search over
  large spatial graphs,'' \emph{PVLDB}, vol.~10, no.~6, pp. 709--720, 2017.

\bibitem{wang2018efficient}
K.~Wang, X.~Cao, X.~Lin, W.~Zhang, and L.~Qin, ``Efficient computing of
  radius-bounded k-cores,'' in \emph{ICDE}.\hskip 1em plus 0.5em minus
  0.4em\relax IEEE, 2018, pp. 233--244.

\bibitem{ghafouri2020efficient}
M.~Ghafouri, K.~Wang, F.~Zhang, Y.~Zhang, and X.~Lin, ``Efficient graph
  hierarchical decomposition with user engagement and tie strength,'' in
  \emph{DASFAA}.\hskip 1em plus 0.5em minus 0.4em\relax Springer, 2020, pp.
  448--465.

\bibitem{zhang2020exploring}
C.~Zhang, F.~Zhang, W.~Zhang, B.~Liu, Y.~Zhang, L.~Qin, and X.~Lin, ``Exploring
  finer granularity within the cores: Efficient (k, p)-core computation,'' in
  \emph{ICDE}.\hskip 1em plus 0.5em minus 0.4em\relax IEEE, 2020, pp. 181--192.

\bibitem{liu2019corecube}
B.~Liu, F.~Zhang, C.~Zhang, W.~Zhang, and X.~Lin, ``Corecube: Core
  decomposition in multilayer graphs,'' in \emph{WISE}.\hskip 1em plus 0.5em
  minus 0.4em\relax Springer, 2019, pp. 694--710.

\bibitem{fang2018effective}
Y.~Fang, Z.~Wang, R.~Cheng, H.~Wang, and J.~Hu, ``Effective and efficient
  community search over large directed graphs,'' \emph{TKDE}, vol.~31, no.~11,
  pp. 2093--2107, 2018.

\bibitem{fang2018spatial}
Y.~Fang, Z.~Wang, R.~Cheng, X.~Li, S.~Luo, J.~Hu, and X.~Chen, ``On
  spatial-aware community search,'' \emph{IEEE Transactions on Knowledge and
  Data Engineering}, vol.~31, no.~4, pp. 783--798, 2018.

\bibitem{zhang2017finding}
F.~Zhang, Y.~Zhang, L.~Qin, W.~Zhang, and X.~Lin, ``Finding critical users for
  social network engagement: The collapsed k-core problem,'' in \emph{AAAI},
  2017, pp. 245--251.

\bibitem{zhang2017olak}
F.~Zhang, W.~Zhang, Y.~Zhang, L.~Qin, and X.~Lin, ``Olak: an efficient
  algorithm to prevent unraveling in social networks,'' \emph{Proceedings of
  the VLDB Endowment}, vol.~10, no.~6, pp. 649--660, 2017.

\bibitem{akbas2017truss}
E.~Akbas and P.~Zhao, ``Truss-based community search: a truss-equivalence based
  indexing approach,'' \emph{PVLDB}, vol.~10, no.~11, pp. 1298--1309, 2017.

\bibitem{zheng2017finding}
Z.~Zheng, F.~Ye, R.-H. Li, G.~Ling, and T.~Jin, ``Finding weighted k-truss
  communities in large networks,'' \emph{Information Sciences}, vol. 417, pp.
  344--360, 2017.

\bibitem{liu2021efficient}
B.~Liu, F.~Zhang, W.~Zhang, X.~Lin, and Y.~Zhang, ``Efficient community search
  with size constraint,'' in \emph{{ICDE}}.\hskip 1em plus 0.5em minus
  0.4em\relax {IEEE}, 2021.

\bibitem{zhang2018finding}
F.~Zhang, C.~Li, Y.~Zhang, L.~Qin, and W.~Zhang, ``Finding critical users in
  social communities: The collapsed core and truss problems,'' \emph{TKDE},
  2018.

\bibitem{yuan2017index}
L.~Yuan, L.~Qin, W.~Zhang, L.~Chang, and J.~Yang, ``Index-based densest clique
  percolation community search in networks,'' \emph{TKDE}, vol.~30, no.~5, pp.
  922--935, 2017.

\bibitem{fang2019efficient}
Y.~Fang, K.~Yu, R.~Cheng, L.~V.~S. Lakshmanan, and X.~Lin, ``Efficient
  algorithms for densest subgraph discovery,'' \emph{PVLDB}, vol.~12, no.~11,
  pp. 1719--1732, Jul. 2019.

\bibitem{garas2012k}
A.~Garas, F.~Schweitzer, and S.~Havlin, ``A k-shell decomposition method for
  weighted networks,'' \emph{New Journal of Physics}, vol.~14, no.~8, p.
  083030, 2012.

\bibitem{eidsaa2013s}
M.~Eidsaa and E.~Almaas, ``S-core network decomposition: A generalization of
  k-core analysis to weighted networks,'' \emph{Physical Review E}, vol.~88,
  no.~6, 2013.

\bibitem{newman2001scientific}
M.~E. Newman, ``Scientific collaboration networks. i. network construction and
  fundamental results,'' \emph{Physical review E}, vol.~64, no.~1, p. 016131,
  2001.

\bibitem{he2020exploring}
Y.~He, K.~Wang, W.~Zhang, X.~Lin, and Y.~Zhang, ``Exploring cohesive subgraphs
  with vertex engagement and tie strength in bipartite graphs,'' \emph{arXiv
  preprint arXiv:2008.04054}, 2020.

\bibitem{liu2020efficient}
B.~Liu, L.~Yuan, X.~Lin, L.~Qin, W.~Zhang, and J.~Zhou, ``Efficient
  ({\(\alpha\)}, {\(\beta\)})-core computation in bipartite graphs,''
  \emph{VLDB J.}, pp. 1--25, 2020.

\bibitem{wang2019vertex}
K.~Wang, X.~Lin, L.~Qin, W.~Zhang, and Y.~Zhang, ``Vertex priority based
  butterfly counting for large-scale bipartite networks,'' \emph{PVLDB},
  vol.~12, no.~10, pp. 1139--1152, 2019.

\bibitem{kobren2019paper}
A.~Kobren, B.~Saha, and A.~McCallum, ``Paper matching with local fairness
  constraints,'' in \emph{SIGKDD}, 2019, pp. 1247--1257.

\end{thebibliography}
}
\end{document}